\documentclass[twoside,11pt]{article}

\usepackage{blindtext}

 \usepackage[abbrvbib, preprint]{jmlr2e}

\usepackage{lastpage}
\jmlrheading{23}{2025}{1-\pageref{LastPage}}{1/21; Revised 5/22}{9/22}{21-0000}{Ogonnaya M. Romanus and Roberto Molinari}

\ShortHeadings{P-COQS: Private Conformity via Quantile Search}{Romanus and Molinari}
\firstpageno{1}

\usepackage{todonotes}  % Allows \todo commands

\usepackage{amsmath, bbm, bm}
\usepackage{graphicx}
\usepackage{subcaption}
\usepackage{float} 
\usepackage{pgfplots}
\usepackage{filecontents}
\usepackage{tikz}
\usepackage{nicefrac}
\usepackage{algorithm}
\usepackage{algpseudocode}
\usepackage{booktabs}
\usepackage{subcaption}
\usepackage{pgfplotstable}
\usepgfplotslibrary{statistics, groupplots}
\pgfplotsset{compat=1.18}
\usepackage{xcolor}
\usepackage[margin=1in]{geometry}

\definecolor{auburnorange}{RGB}{232, 119, 34}
\definecolor{auburnblue}{RGB}{12, 35, 64}
\definecolor{lightgray}{gray}{0.9}
\definecolor{lightbluegray}{RGB}{220,230,240}

\newtheorem{remk}{Remark}
\newtheorem{prop}{Proposition}

\definecolor{nonprivcolor}{RGB}{0,114,189} % blue
\definecolor{anascolor}{RGB}{217,83,25} % red
\definecolor{privquantcolor}{RGB}{119,172,48} % green

\begin{document}

\title{Differentially Private Conformal Prediction\\ via Quantile Binary Search}

\author{\name Ogonnaya Michael Romanus \email omr0010@auburn.edu \\
       \addr Department of Mathematics \& Statistics\\
       Auburn University\\
       Auburn, AL 36849, USA
       \AND
       \name Roberto Molinari \email robmolinari@auburn.edu \\
       \addr Department of Mathematics \& Statistics\\
       Auburn University\\
       Auburn, AL 36849, USA}

\editor{My editor}

\maketitle

\begin{abstract}%   <- trailing '%' for backward compatibility of .sty file
Differentially Private (DP) approaches have been widely explored and implemented for a broad variety of tasks delivering corresponding privacy guarantees in these settings. While most of these DP approaches focus on limiting privacy leakage from learners based on the data that they are trained on, there are fewer approaches that consider leakage when procedures involve a calibration dataset which is common in uncertainty quantification methods such as Conformal Prediction (CP). Since there is a limited amount of approaches in this direction, in this work we deliver a general DP approach for CP that we call Private Conformity via Quantile Search (P-COQS). The proposed approach adapts an existing randomized binary search algorithm for computing DP quantiles in the calibration phase of CP thereby guaranteeing privacy of the consequent prediction sets. This however comes at a price of slightly under-covering with respect to the desired $(1 - \alpha)$-level when using finite-sample calibration sets (although broad empirical results show that the P-COQS generally targets the required level in the considered cases). Confirming properties of the adapted algorithm and quantifying the approximate coverage guarantees of the consequent CP, we conduct extensive experiments to examine the effects of privacy noise, sample size and significance level on the performance of our approach compared to existing alternatives. In addition, we empirically evaluate our approach on several benchmark datasets, including CIFAR-10, ImageNet and CoronaHack. Our results suggest that the proposed method is robust to privacy noise and performs favorably with respect to the current DP alternative in terms of \textit{empirical coverage}, \textit{efficiency}, and \textit{informativeness}. Specifically, the results indicate that P-COQS produces smaller conformal prediction sets while simultaneously targeting the desired coverage and privacy guarantees in all these experimental settings.
\end{abstract}

\begin{keywords}
  Privacy, Quantile Estimation, Computational Efficiency, Uncertainty Quantification
\end{keywords}

\section{Introduction}

In high-stakes domains such as health care, criminal justice, and financial risk assessment, machine learning models are increasingly being used to guide decision-making processes. However, traditional approaches often produce single-point predictions without quantifying uncertainty, which could lead to overconfident and unreliable outcomes. Many uncertainty quantification techniques have been proposed to address this need, including (among others) Bayesian neural networks, ensemble methods, and dropout-based approximations \citep[see e.g.][]{blundell2015weight, lakshminarayanan2017simple, osband2016deep, gal2016dropout, kendall2017uncertainties}. While these methods can offer meaningful uncertainty estimates, they often rely on strong modeling assumptions, are computationally intensive, or lack rigorous coverage guarantees (especially under model misspecification or when applied out-of-distribution). Conformal Prediction (CP) addresses various of these limitations by generating \emph{prediction sets}, that is, subsets of the response space $\mathcal{Y}$ that, under the assumption of exchangeability, are guaranteed to contain the true response with a user-specified probability $1 - \alpha$, where $\alpha \in (0,1)$. These sets provide a reliable and model-agnostic framework for uncertainty quantification, making them increasingly attractive in safety-critical applications \citep[see e.g.][]{uddin2023applications, astigarraga2025conformal, tuwani2023safe, olsson2022estimating}. While conformal methods require i.i.d. data for exact guarantees, recent extensions have begun to explore how to adapt them to distribution shift settings with approximate validity \citep[see e.g.][]{tibshirani2019conformal, gibbs2021adaptive, jonkers2024wcps}. However, while CP provides statistical reliability, its reliance on calibration data introduces privacy risks, as any publicly released statistic based on these data could be exploited by an adversary and used for malicious purposes. Moreover, when training models on data that collect individual-level information, these models tend to preserve individual-specific information, and some of them, such as deep neural networks, are known to \emph{memorize} training examples \citep{fredrikson2015model,carlini2021extracting}, enabling adversaries to extract sensitive information via model inversion \citep{fredrikson2015model} or membership inference attacks \citep{shokri2017membership}. This poses a significant challenge also in regulated domains (e.g., under GDPR, HIPAA, or CCPA), where models must not leak information about individuals in the training or calibration sets.

To address these issues, Differential Privacy (DP) \citep{dwork2006differential} has emerged as one of the main frameworks for privacy-preserving machine- and statistical-learning, ensuring that model outputs do not reveal whether any individual's data was included in the training set. However, existing DP frameworks focus primarily on protecting the \emph{training} phase, leaving the calibration phase (and thus the CP step) vulnerable to privacy breaches. For example, an attacker could exploit the calibration scores to infer whether a specific individual was included in the calibration data \citep{c1}. Given this problem, there are nevertheless very few existing approaches that ensure that CP procedures can provide DP guarantees while maintaining reasonable performance in terms of uncertainty quantification metrics and computational efficiency. Hence, this work aims to further contribute to the goal of uncertainty quantification under privacy constraints by adapting an existing DP quantile selection approach based on binary search \citep{c2} and using it as a subroutine to deliver CP sets that guarantee privacy. We confirm the theoretical properties of the resulting CP procedure and highlight its favorable experimental performance compared to the few existing alternatives, both in terms of uncertainty quantification and computational efficiency.

\section{Related Work}
\label{sec:rel_work}

The idea behind CP was first described in \citet{c3} followed by different papers that developed its theory and applications (see \citet{c6}, \citet{c9}, \citet{c10}, \citet{c11}, and the references therein). Based on these, considerable attention has been paid in recent years to adapt CP to different settings and deliver additional applied and theoretical results (see \citet{c39} for a recent detailed review). Despite these different research outputs, little attention has been devoted to guaranteeing privacy for the calibration data through protection of the derived \textit{non-conformity scores} used to deliver the prediction sets for uncertainty quantification. The only existing work related to the central DP approach we present in this paper is that of \citet{c1} where they privatize the CP procedure by discretizing the non-conformity scores into bins and applying the \textit{exponential mechanism} (\citet{c40}) to select a DP quantile which is then used to form the prediction sets. More specifically, to maintain coverage guarantees under $\epsilon$-DP, the target quantile is adjusted (inflated) to account for added noise. This ensures both privacy (via DP) and valid coverage, albeit with potentially larger prediction sets due to the privacy-utility trade-off. A possible practical limitation of this proposed approach lies in the choice of the optimal number of bins, as well as a hyperparameter that needs to be tuned to minimize the inflated quantile to achieve the smallest possible prediction set. Although procedures have been determined to compute these in an optimal manner, they can potentially increase the time complexity of this method. In recent years there have been few other CP approaches proposed in the setting of federated learning (see \cite{humbert2023one}) and \textit{local DP} \cite{penso2025privacy}, however, these do not fit in the central DP framework of this paper.

\section{Preliminaries}

\subsection{Conformal Prediction}
\label{sec:conf_pred}

To define the framework of the proposed method, let us consider $n$ exchangeable data points $\mathcal{D} = \{(X_1,Y_1), \dots, (X_n, Y_n)\}$ with $X_i \in \mathcal{X}^d$ being a $d$-dimensional input vector and $Y_i \in \mathcal{Y}$ being the target variable. Using the basic \emph{split} CP \citep{c4,c5}, we can partition these points into a training set with $n_{\text{train}}$ observations (to fit a model of choice $f$) and a calibration set with $n_{\text{cal}} = n - n_{\text{train}}$ observations to quantify prediction uncertainty (assuming test data is taken from another sample). If we let $\mathcal{I} = \{1, \hdots, n\}$ represent the index set of the observations, then we can define $\mathcal{I}_{\text{cal}} \subset \mathcal{I}$ as the index set of observations in the calibration set. Assuming that we have trained a model $\hat{f}$ on the training set, the goal of CP is to compute \textit{non-conformity scores} $s_i = s(X_i, Y_i, \hat{f})$ for each calibration point $(X_i, Y_i), \, \forall i \in \mathcal{I}_{\text{cal}}$. These scores measure the discrepancy between the model's prediction and the true label such that lower scores suggest better agreement, that is, better predictions from the trained model $\hat{f}$. Therefore we now have a set of non-conformity scores $\mathcal{S} =\{s_1, \hdots, s_{\text{cal}}\}$ on which we can obtain $\alpha$-quantiles with $\alpha \in (0,1)$ being the significance level which represents the amount of errors (in percentage) that the model is allowed to make. We denote this quantile as $q$ and let its dependence on the significance level $\alpha$ be implicit. Once the scores are computed, we can therefore sort them in increasing order (that is,  $\{s_{[1]} \leq \hdots \leq s_{[n_{\text{cal}}]}\}$) and estimate this quantile by selecting the score in the $r^{\text{th}}$ position where we define $r = \lceil (1-\alpha) (n_{\text{cal}}+1) \rceil$, with $\lceil v \rceil$ representing the smallest integer greater than $x$, leading to $\hat{q} = s_{[r]}$. Suppose that we now have a new test data point $X_{\text{test}}$,  the resulting prediction set for the target $Y_{\text{test}}$ is defined as:
\[
\mathcal{C}(X_{\text{test}}) = \left\{ Y \in \mathcal{Y} : s(X_{\text{test}}, Y, \hat{f}) \leq \hat{q} \right\}.
\] 
Under the assumption of exchangeability of the data points, this set guarantees the coverage probability $\mathbb{P}(Y_{\text{test}} \in \mathcal{C}(X_{\text{test}})) \geq 1-\alpha$, where the probability is taken over the $n+1$ data points $ (X_1,Y_1), \dots, (X_n, Y_n), (X_{\text{test}},Y_{\text{test}})$. For further details and extensions we refer the reader to \citet{c39}.

\subsection{Differential Privacy}
\label{sec:diff_priv}

Differential Privacy (DP) currently represents one of the main frameworks for privacy protection with mathematically provable guarantees. More in detail, a randomized mechanism $\mathcal{M}$ is said to be $(\epsilon, \delta)$-DP if, for all neighboring datasets $\mathcal{D}$ and $\mathcal{D'}$ (that is, datasets with Hamming distance $d_{ham} (\mathcal{D},\mathcal{D'})=1$) and for all measurable subsets $\mathcal{Q}$ of the output space, we have that
\[
\mathbb{P}[\mathcal{M}(\mathcal{D}) \in \mathcal{Q}] \leq e^{\epsilon} \mathbb{P}[\mathcal{M}(\mathcal{D}') \in \mathcal{Q}] +\delta.
\]
There are many variations of this framework, including Gaussian-DP which benefits from a series of nice additional properties compared to the above $(\epsilon, \delta)$-DP \citep[][]{ dwork2006differential, dong2022gaussian}. For the purpose of this work, we rely on a specific version of DP defined below.

\begin{definition}[Zero-Concentrated Differential Privacy,  \cite{bun2016}]
For $\rho >0$, a randomized mechanism $\mathcal{M}$ satisfies $(\rho, 0)$-Zero-Concentrated Differential Privacy (zCDP) if for all neighboring datasets $\mathcal{D}$ and $\mathcal{D}'$ and for all $\alpha > 1$, the Rényi divergence $D$ of order $\alpha$ satisfies:
\[
D_{\alpha}(\mathcal{M}(\mathcal{D}) \| \mathcal{M}(\mathcal{D}')) \leq \rho \alpha.
\]
\end{definition}

Unlike $(\epsilon, \delta)$-DP, $\rho$-zCDP is \textit{tight} under sequential composition without the need for the \textit{advanced composition theorem}, which may consume extra privacy budget. Nevertheless, there is a direct relationship between the $(\epsilon, \delta)$-DP framework and the zCDP framework: indeed, if $\mathcal{M}$ satisfies $\rho$-zCDP, then for $\delta>0$, $\mathcal{M}$ satisfies $(\epsilon, \delta)$-DP for $\epsilon=\rho +2\sqrt{\rho \log(1/\delta)}$ \citep[see][]{near2021programming}. Hence, the latter implies the former. In particular, a randomized mechanism $\mathcal{M}$ satisfies $(\epsilon,0)$-DP if and only if it satisfies $\epsilon$-zCDP \citep{bun2016}.

\section{P-COQS: Private Conformity via Quantile Search}

% \textcolor{red}{CHECK TO MAKE SURE WE USE THE CORRECT NOTATION BETWEEN $Y$ AND $y$ or $y_{\text{test}}$ THROUGHOUT THE PAPER. I ASSUME THE RANDOM VARIABLE IS THE CONFIDENCE SET $\mathcal{C}$ AND $y$ IS FIXED?} \textcolor{blue}{DONE.}

As highlighted in \cite{c1}, the release of the non-conformity score quantile or of the prediction sets based on this quantile can compromise the privacy of individuals in the calibration set. Therefore, a solution to this problem is to directly release a DP quantile or to release the corresponding prediction sets, which would also preserve DP based on the privatized quantile. While the existing alternative in \cite{c1} employs the exponential mechanism to sample a DP quantile, this work relies on a direct adaptation of the DP binary search algorithm presented in \cite{c2}. More specifically, the latter algorithm releases a DP quantile by employing a noisy count function (denoted as \texttt{NoisyRC}) within a binary search procedure run on a sequence of ordered integers. In particular, the function \texttt{NoisyRC}$([a,b], \mathcal{D})$ returns a noisy count of $\mathcal{D} \cap [a,b]$, where $\mathcal{D}$ is a set containing ordered integers that are considered sensitive and hence must be accessed privately \citep[see][]{c2}. However, since in this work our sequences are not integers (given that the non-conformity scores are continuous), we adapt this DP quantile selection approach in a straightforward manner and, for completeness, we reproduce their adapted approach in Algorithm \ref{alg:P-COQS} with some modification to suit our application. 

\begin{algorithm}
\caption{Adapted \texttt{PrivQuant} Algorithm of \cite{c2}}
\label{alg:P-COQS}
\begin{algorithmic}[1]
\Require Non-conformity scores $\mathcal{S} = \{s_{1}, \dots, s_{n_{\text{cal}}}\}$, significance level $\alpha \in (0,1)$, lower and upper bounds for non-conformity scores $[a,b]$, $\delta > 0$ (small positive value, default $\delta = 10^{-10}$), privacy parameter $\rho$
\Ensure DP quantile $q^{\text{DP}}$
\State Fix $r = \lceil (1-\alpha)(n_{\text{cal}}+1) \rceil$ and $N = \lceil \log_2\left(\nicefrac{b-a}{\delta}\right) \rceil$
\State Initialize: $\text{left} \gets a$, $\text{right} \gets b$, $i \gets 0$
\While{$i \leq N$}
    \State $\text{mid} \gets \frac{\text{left} + \text{right}}{2}$
    \State $c \gets \texttt{NoisyRC}([a, \text{mid}], \mathcal{S})$
    \If{$c < r$}
        \State $\text{left} \gets \text{mid} + \delta$
    \Else
        \State $\text{right} \gets \text{mid}$
    \EndIf
    \State $i \gets i + 1$
\EndWhile
\State \Return $q^{\text{DP}} = \frac{\text{left} + \text{right}}{2}$
\end{algorithmic}
\end{algorithm}

Broadly speaking, Algorithm \ref{alg:P-COQS} performs a binary search over the interval $[a, b]$ (often corresponding to $[0, 1]$ in the case of classification tasks) by iteratively narrowing the search space to localize the quantile of interest while preserving privacy through a randomized counting mechanism. More specifically, at each iteration, the algorithm evaluates a midpoint and calls the function \texttt{NoisyRC}, which returns a DP count of the number of non-conformity scores (in $\mathcal{S}$) less than or equal to the midpoint. This count guides the search by indicating whether the true quantile lies to the left or right of the midpoint. The process continues until it reaches iteration $N$ which is the maximum number of iterations allowed to ensure $\rho$-zCDP. Indeed, as stated further on, the DP noise used for the function \texttt{NoisyRC} needs to be scaled to the number of iterations in Algorithm \ref{alg:P-COQS}. Since we would ideally want the algorithm to terminate when the interval length is smaller than a small value $\delta$ (i.e. $\text{right} - \text{left} \leq \delta$), the fixed number of iterations $N$ needed to reach this point is such that $\nicefrac{b - a}{2^N} \leq \delta$ (since the interval is divided by two after each iteration) which, solving for $N$, gives us
$$N = \lceil \log_2\left(\nicefrac{b-a}{\delta}\right) \rceil.$$ 
With this representing the exact number of iterations to guarantee the required condition (i.e. $\text{right} - \text{left} \leq \delta$), Algorithm \ref{alg:P-COQS} returns the midpoint as the private quantile estimate $q^{\text{DP}}$. Once this quantile is obtained, following the definition in Section \ref{sec:conf_pred}, the DP conformal prediction set is simply defined as:
\begin{equation}
\label{eq:cops}
 \mathcal{C}^{DP}(X_{\text{test}}) = \left\{ Y \in \mathcal{Y} : s(X_{\text{test}}, Y, \hat{f}) \leq q^{DP} \right\}.
\end{equation}
We refer to the above procedure as Private Conformity via Quantile Search (P-COQS) which therefore is built on the adaptation described in Algorithm \ref{alg:P-COQS} and the consequent use of the resulting DP quantile $q^{DP}$ in Equation \eqref{eq:cops}. The entire P-COQS procedure is summarized in Algorithm \ref{alg:P-COQS_complete}.

\begin{algorithm}
\caption{P-COQS}
\label{alg:P-COQS_complete}
\begin{algorithmic}[1]
\Require Training data $\mathcal{D}_{\text{train}} = \{(X_1, Y_1), \dots, (X_{n_{\text{train}}}, Y_{n_{\text{train}}})\}$, calibration data $\mathcal{D}_{\text{cal}} = \{(X_1, Y_1), \dots, (X_{n_{\text{cal}}}, Y_{n_{\text{cal}}})\}$, test data point $X_{\text{test}}$, significance level $\alpha \in (0,1)$, lower and upper bounds for non-comformity scores $[a,b]$, $\delta > 0$ (small positive value, default $\delta = 10^{-10}$), privacy parameter $\rho$
\Ensure DP prediciton set $\mathcal{C}^{\text{DP}}$
\State Train model on training data $\mathcal{D}_{\text{train}}$ to obtain $\hat{f}$
\State Obtain non-conformity scores $\mathcal{S} = \{s_{1}, \dots, s_{n_{\text{cal}}}\}$ from calibration data $\mathcal{D}_{\text{cal}}$ using $\hat{f}$
\State Run Algorithm \ref{alg:P-COQS} to obtain $q^{DP}$
\State Compute $\mathcal{C}^{DP}$ based on Equation \eqref{eq:cops} using $q^{\text{DP}}$
\State \Return $\mathcal{C}^{DP}$
\end{algorithmic}
\end{algorithm}

Let us now focus on the properties of P-COQS and, for completeness, we start by studying the error bound and DP properties of Algorithm \ref{alg:P-COQS} which directly follow from the properties of the \texttt{PrivQuant} algorithm in \cite{c2}. Indeed, while the latter algorithm is tailored to search over positive integers, as mentioned earlier, Algorithm \ref{alg:P-COQS} adapts the latter to any ordered data type and any set of bounds $[a, b]$ on the real line. Let us therefore discuss (confirm) the DP properties of this adaptation and, to do so, let us first define the noisy count function as
\begin{equation}
\label{eq:zcdp_count}
\texttt{NoisyRC}([a,\text{mid}],\mathcal{S})= \texttt{card}\big(\mathcal{S} \cap [a,\text{mid}] \big) + \mathcal{N}\left(0, \frac{\lceil \log_2(\nicefrac{b-a}{\delta})\rceil}{2\rho}\right),
\end{equation}
where $\rho > 0$ is the privacy parameter under zCDP, with the sensitivity scaled to the maximum number of iterations $N$ for Algorithm \ref{alg:P-COQS} (considering that count query $l_1$-sensitivity is 1), and $\texttt{card}$ denotes the cardinality of set. Let us now denote $u = \nicefrac{b - a}{\delta}$ and let $\Phi$ represent the standard normal CDF. Then, adapting directly from \cite{c2}, the following proposition holds.

\begin{prop}
\label{prop2}
Algorithm \ref{alg:P-COQS} using the noisy count function in \eqref{eq:zcdp_count} is $\rho$-zCDP and returns a quantile with rank error
\[
\tau = \sqrt{\frac{\lceil \log_2 (u)\rceil}{\rho}\log\left(2\frac{\lceil \log_2 \big(u\big)\rceil}{\beta} \right)}
\]
with probability at least $1 - \beta$.
\end{prop}

\begin{proof}
The proof is straightforward. Indeed, since the sensitivity of the range query $\texttt{card}\big(\mathcal{S} \cap [a, \text{mid}]\big)$ is 1, then adding noise sampled from $\mathcal{N}\big(0, \nicefrac{\lceil \log_2(u)\rceil}{2\rho}\big)$ will guarantee $\nicefrac{\rho}{\lceil \log_2(u)\rceil}$-zCDP for each call of \texttt{NoisyRC}$([a,\text{mid}],\mathcal{S})$. Since there are $N = \lceil \log_2(u)\rceil$ calls before the algorithm terminates, then by the composition property of zCDP \citep{bun2016}, Algorithm \ref{alg:P-COQS} is $\rho$-zCDP.  Now let $\sigma^2 =\nicefrac{N}{2\rho}$ denote the variance of the Gaussian noise added to guarantee $\rho$-zCDP in \eqref{eq:zcdp_count}. Hence, the deviation between the true count $\texttt{card}(\mathcal{S}\cap [a,\text{mid}])$ and the noisy count $\texttt{NoisyRC}([a,\text{mid}],\mathcal{S})$, which we respresent as the random variable $W$, follows a Gaussian distribution $\mathcal{N}(0,\sigma^2)$. We want to bound the error for the $N$ calls with high probability $1-\beta$, for $\beta\in (0,1)$. For each call of the algorithm, we have

\begin{align*}
    \mathbb{P}(\big| \text{NoisyRC}([L,\text{mid}],D) - \texttt{card}\big(D \cap [L,\text{mid}] \big)\big|> \tau)=\mathbb{P}(|W|>\tau).
\end{align*}

\noindent Since there are exactly $N$ calls before the algorithm terminates, then by union bound and Gaussian tail probability we have that:
$$ \mathbb{P}\Big(\bigcup_{i=1}^N\{|W_i|>\tau\}\Big)\leq N \mathbb{P}\big(|W|>\tau\big)\leq 2Ne^{-\frac{\tau^2}{2\sigma^2}}.$$

\noindent To ensure that this probability be at most $\beta$, we thus estimate the error bound as 
$$2Ne^{-\frac{\tau^2}{2\sigma^2}}\leq \beta \implies  \tau \geq \sigma \sqrt{2 \ln\left(\frac{2N}{\beta}\right)}.$$
Substituting the value of $\sigma$ and $N$ and simplifying we obtain:
$$\tau \geq  \sqrt{\frac{\lceil \log_2 (u)\rceil}{\rho}\log\left(2\frac{\lceil \log_2 \big(u\big)\rceil}{\beta} \right)}.$$

% \noindent Therefore, with probability at least $1-\beta$ we have that the algorithm's rank error is bounded by $$\tau =  \sqrt{\frac{\lceil \ln (u)\rceil}{\rho}\ln\left(\frac{2\lceil \ln (u)\rceil}{\beta} \right)}.$$

% \begin{equation*}
%     \mathbb{P}\left(\big| \texttt{NoisyRC}([a,\text{mid}],\mathcal{S}) - \texttt{card}\big(\mathcal{S} \cap [a, \text{mid}]\big)\big|\leq \tau \right)=\mathbb{P}\left(|\mathcal{N}(0,\sigma^2)|\leq \tau \right) \geq 1-\beta,
% \end{equation*}
% %
% Thus, being the above the definition of a quantile, the rank error bound $\tau$ is given by:
% %
% \begin{equation*}
%     \tau =\sigma \, \Phi^{-1}\left(1-\nicefrac{\beta}{2}\right) = \sqrt{\frac{\lceil \log_2(\nicefrac{b-a}{\delta})\rceil}{2\rho}} \, \Phi^{-1}\left(1-\frac{\beta}{2}\right).
% \end{equation*}
\end{proof}

\begin{remk}
    Assuming $\lceil \log_2(u) \rceil = \log_2(u)$ and using change-of-base we have that
    \begin{equation*}
        \sqrt{\frac{\lceil \log_2 (u)\rceil}{\rho}\log\left(\frac{2\lceil \log_2 \big(u\big)\rceil}{\beta} \right)} = \sqrt{\frac{\log(u)}{\log(2) \rho} \log \left( \frac{2\log(u)}{\log(2) \beta} \right)}.
    \end{equation*}
    Looking at the inner logarithm we have
    $$\log\left( \frac{2 \log(u) }{ \beta \log(2) } \right)
= \log\left( \frac{ \log(u) }{ \beta } \right) + \log\left( \frac{2}{\log (2)} \right),$$
which, using asymptotic approximations, leads to
$$\tau \approx \frac{ \log (u) }{ \rho \log (2) } \cdot \log\left( \frac{ \log (u) }{ \beta } \right),$$
since the second term is constant and the first term therefore dominates in $u$. This gives the same polylogarithmic order of $\tau$ as in \cite{c2} which is near optimal.
\end{remk}

% \begin{remk}
% If one wanted a more conservative value of $\tau$, then we could use concentration tail bounds for Gaussian noise as follows:
% %
% \begin{equation*}
%    \mathbb{P}(|\mathcal{N}(0,\sigma^2)|>\tau) \,  \leq \, 2\, \exp\left(-\frac{\tau^2}{2\sigma^2}\right).
% \end{equation*}
% %
% Wanting this probability to be at most $\beta$ allows us to obtain the following:
% %
% \begin{equation*}
%    2\,\exp\left(-\frac{\tau^2}{2\sigma^2}\right)\leq \beta \iff -\frac{\tau^2}{2\sigma^2} \leq \ln \left(\frac{\beta}{2}\right) \iff \tau \geq \sigma \, \sqrt{2\ln \left(\frac{2}{\beta}\right)}.
% \end{equation*}
% %
% Therefore, we could define 
% $$\tau = \sqrt{\frac{\lceil \log_2(\nicefrac{b-a}{\delta})\rceil \ln{(\nicefrac{2}{\beta})} }{\rho}},$$
% %
% which obviously delivers a looser bound compared to that of Proposition \ref{prop2} since it ensures the required condition even if the actual probability is much smaller.
% \end{remk}
\noindent
These results allow us to quantify to what extent the rank associated with the private quantile $q^{DP}$ from Algorithm \ref{alg:P-COQS} differs from the non-private rank in the calibration data. Firstly, obtaining a $\rho$-zCDP quantile ensures that Algorithm \ref{alg:P-COQS_complete} preserves the desired privacy level. Moreover, Proposition \ref{prop2} allows us to explicitly define this difference $\tau$ with high probability (at least $1 - \beta$). We will see that this information can be used and is indeed helpful when determining the coverage guarantees of P-COQS in Algorithm \ref{alg:P-COQS_complete} which are stated in Theorem \ref{thm:cov_guarantee}. For this purpose, we also define $s_{\text{test}}~=~s(X_{\text{test}}, Y, \hat{f})$.

\begin{theorem}[Coverage Guarantee]
\label{thm:cov_guarantee}
Let $(s_1, \dots, s_{n_{\text{cal}}}, s_{\text{test}})$ be an exchangeable sequence of non-conformity scores. Then, the P-COQS prediction set $\mathcal{C}^{DP}(X_{\text{test}})$ in Algorithm \ref{alg:P-COQS_complete} satisfies:
\[
1 - \alpha - \frac{\tau}{n_{\text{cal}}+1} \leq \mathbb{P}\left[Y_{\text{test}} \in \mathcal{C}^{DP}(X_{\text{test}})\right] \leq 1 - \alpha + \frac{\tau + 1}{n_{\text{cal}}+1}.
\]
\end{theorem}

\begin{proof}
 By Proposition \ref{prop2} we have that Algorithm \ref{alg:P-COQS} outputs a private quantile $q^{DP}$ with rank error $\tau$ ensuring that $s_{r-\tau}\leq s_{r} \leq s_{r+\tau}$. By the exchangeability condition of the non-conformity scores, the rank of $s_{\text{test}} = s(X_{\text{test}}, Y, \hat{f})$ among $(s_1, \dots, s_{n_{\text{cal}}}, s_{\text{test}})$ is uniformly distributed over $\{1, \hdots, n_{\text{cal}}+1\}$. Thus, recalling that $r = \lceil (1-\alpha)(n_{\text{cal}}+1) \rceil$ and that 
 $$Y\in \mathcal{C}^{DP}(X_{\text{test}}) \iff s_{\text{test}} \leq q^{DP},$$
 we consequently have:
\begin{equation*}
    \underbrace{\frac{r -\tau}{n_{\text{cal}}+1}}_{L} \leq \mathbb{P}\left(s_{\text{test}} \leq q^{DP}\right) \leq \underbrace{\frac{r +\tau}{n_{\text{cal}}+1}}_{U},
\end{equation*}
where:
\begin{align*}
    L &= \frac{\lceil (1-\alpha)(n_{\text{cal}}+1) \rceil - \tau}{n_{\text{cal}}+1} \geq 1-\alpha - \frac{\tau}{n_{\text{cal}}+1} \,\,\,\,\,\, \text{and}\\
    U &= \frac{\lceil (1-\alpha)(n_{\text{cal}}+1) \rceil + \tau}{n_{\text{cal}}+1} \leq \frac{(1-\alpha)(n_{\text{cal}}+1) +1 +\tau}{n_{\text{cal}}+1} = 1-\alpha + \frac{\tau+1}{n_{\text{cal}}+1}.
\end{align*}
Therefore this allows us to conclude that
\begin{equation*}
  1-\alpha - \frac{\tau}{n_{\text{cal}}+1} \leq \mathbb{P}\left[Y_{\text{test}} \in \mathcal{C}^{DP}(X_{\text{test}}) \right] \leq 1-\alpha + \frac{\tau+1}{n_{\text{cal}}+1}.
\end{equation*}
\end{proof}

Although Theorem \ref{thm:cov_guarantee} does not ensure the coverage probability $(1 - \alpha)$, as opposed to the method in \cite{c1} that preserves this standard CP guarantee, it establishes an approximate coverage bound with an error term of order $\mathcal{O}(\nicefrac{\tau}{n_{\text{cal}}})$. While this indeed implies an error margin for the actual coverage, it nevertheless also provides an upper bound to the coverage guarantee which allows to understand how conservative and how efficient the resulting prediction sets can be \citep[see e.g.][]{lei2018distribution, li2025conformal}. More in detail, this error depends not only on the size of the calibration set but also on the width of the interval $[a, b]$ and the privacy level $\rho$ under zCDP. Specifically, as shown in Proposition~\ref{prop2}, the coverage error increases with the length of the interval and decreases with the privacy budget. Therefore, achieving accurate coverage requires a sufficiently large calibration set when the parameters of the error bound $\tau$ are fixed. Nevertheless, all the experimental results in Section~\ref{sec:experiments} indicate that P-COQS is able to achieve coverage comparable to the non-private conformal prediction method, despite this weaker theoretical guarantee.

% where, aside from requiring the size of the calibration data to increase, it can be observed from Proposition \ref{prop2} that the error is directly proportional to the size of the interval $[a, b]$ and inversely proportional to the privacy level of zCDP given by $\rho$ (i.e. the smaller the privacy budget, the larger the coverage error margin). In this sense, having fixed the parameters of the error bound $\tau$, this result requires the size of the calibration set $n_{\text{cal}}$ to be moderately large. However, despite this, all the experimental settings in Section \ref{sec:experiments} highlight how P-COQS actually has a coverage performance in line with the non-private CP approach which provides this guarantee. 

% \textcolor{blue}{WE NEED TO TAKE A CLOSER LOOK AT THE ABOVE PARAGRAPH. The experimental results we obtained do not align with these statements. Our method works better than that of Anas et al. in small sample sizes based on coverage and other metrics. A closer look at the algorithm indicates that it does not depend on the sample size of the calibration set: it is just a binary search. Our simulation result depicts the same pattern. Furthermore, in the paper where the original algorithm was proposed, it was stated that the rank error is nearly optimal. But I get your point. The bounds say otherwise.} \textcolor{red}{DONE. KINDLY TAKE A LOOK.}
\begin{remk}
    As highlighted by \citet{c2}, the quantile method adapted in Algorithm \ref{alg:P-COQS} is nearly optimal, implying that our results are suboptimal (polylogarithmic) in terms of CP guarantees. Nevertheless, using the bounds in Proposition \ref{prop2} and following remarks, it delivers a practical way to determine the true differentially private coverage guarantee with high probability (i.e. $1 - \beta$) without having to inflate the coverage to mitigate the impact of privacy noise (which can negatively impact the informativeness of the CP sets). For example, suppose that we are dealing with a classification task where the non-conformity scores lie in the interval $[0,1]$, then by fixing the other parameters $\delta$ and $\rho$ (e.g., $\delta = 10^{-10}$ and $\rho = 0.1$) and choosing $\beta = 0.01$, following Proposition \ref{prop2} the true probability of DP coverage of P-COQS is at least $1 - \alpha - 0.01503$ and at most $1 - \alpha + 0.0154$ guaranteed with probability at least $1 - \beta = 0.99$ and calibration size $n_{\text{cal}} = 3000$. %\textcolor{red}{PLEASE CHECK THE CORRECTNESS OF ALL OF THIS}. \textcolor{blue}{This is now okay. KINDLY TAKE A LOOK. Also, here is the R code:}
% \begin{verbatim}
% b <- 1
% a <- 0
% delta <- 1e-10  
% beta <- 0.01
% rho <- 1
% n <- 3000
% alpha <- 0.1    

% # Expression of tau using the concentration bound for a Gaussian random variable
% tau = sqrt(ceiling(log2((b-a)/delta))*log(2/beta)/rho)
% upper_bound = 1 - alpha + (tau+1)/(n+1)
% lower_bound = 1 - alpha - tau/(n+1)
% c(tau/(n+1), (tau+1)/(n+1))
% c(lower_bound,upper_bound)

% # Expression of tau using the exponential tail bound for Gaussian noise
% tau2 <- sqrt(ceiling(log2((b-a)/delta))/(2*rho))*qnorm(1-beta/2)
% upper_bound2 = 1 - alpha + (tau2+1)/(n+1)
% lower_bound2 = 1 - alpha - tau2/(n+1)
% c(tau2/(n+1), (tau2+1)/(n+1))
% c(lower_bound2,upper_bound2)

% \end{verbatim}
\end{remk}

%
%For example, in classification tasks the non-conformity scores often have natural bounds (i.e.$[0,1]$) thereby simplifying the definition of the search region for Algorithm \ref{alg:P-COQS}. Therefore, while P-COQS applies to both classification and regression settings, for our simulation and applied experiments in the following section we focus on classification.

\section{Experiments}
\label{sec:experiments}

% Since these approaches are based on $(\epsilon, \delta)$-DP, throughout all experiments we employ pure $\epsilon$-DP implying that the results guarantee $\rho$-zCDP with $\rho = \nicefrac{\epsilon^2}{2}$. In particular, we will also study the sensitivity of these approaches to different levels of DP and sample sizes.

%\textcolor{red}{IS THIS CORRECT?} \textcolor{blue}{That is not what we did in the updated version. The DNN were trained as $(\epsilon, \delta)$-DP, while the conformal prediction is based on $\rho$-zCDP, which is equivalent to $\epsilon$-DP.}

% As in \cite{c1}, we limit ourselves to the study of classification tasks and, for this purpose, we use the non-conformity score represented by the \textit{hinge loss}:

In this section, we evaluate the proposed CP approach (\textsc{P-COQS}) under the same experimental setting as~\cite{c1}. Additionally, we extend this evaluation with controlled simulation experiments. Specifically, we compare \textsc{P-COQS} with the method introduced by~\citet{c1} (henceforth referred to as \textsc{ExponQ}), examining scenarios where both approaches employ either non-private models ($\hat{f}$) or DP models ($\hat{f}^{\mathrm{DP}}$). For differentially private models, we train Na\"ive Bayes and Random Forest models using IBM's \texttt{Diffprivlib} library in Python,\footnote{DP implementations based on IBM's Diffprivlib: \url{https://diffprivlib.readthedocs.io/en/latest/modules/models.html\#classification-models}} ensuring pure $\epsilon$-DP guarantees. For deep neural networks, we adopt the same methodology as~\cite{c1}, utilizing the \texttt{Opacus} library to obtain $(\epsilon, \delta)$-DP models. For DP conformal prediction methods, we leverage the exact equivalence between $\epsilon$-DP and $(\epsilon,0)$-zCDP established in Lemma~3.2 of~\cite{bun2016}. This allows us to analyze both \textsc{P-COQS} and \textsc{ExponQ} under either privacy framework while maintaining identical privacy guarantees. We specifically examine how these methods behave under varying privacy parameters ($\epsilon$) and different sample sizes. Consistent with the experimental design of~\cite{c1}, we restrict our investigation to classification tasks. The non-conformity measure we used is given by the \textit{hinge loss} function:
\[
s(x, y, \hat{f}) = 1 - \hat{f}(x)_y,
\]
where $\hat{f}(x)_y$ represents the predicted probability of the model for the true class $y$. We evaluate the performance of CP methods using the metrics defined below. 

\begin{definition}[Prediction Set Quality Metrics]
Given test points $\{(x_i,y_i)\}_{i=1}^{n_{\text{test}}}$ with corresponding  CP sets $\{\mathcal{C}(x_i)\}_{i=1}^{n_{\text{test}}}$:
\begin{itemize}
  \item  \textbf{(Marginal) Coverage}: proportion of prediction sets that cover the true label:
\[
\frac{1}{n_{\text{test}}} \sum_{i=1}^{n_{\text{test}}} \mathbbm{1} \{ y_i \in \mathcal{C}(x_i) \}
\]

    \item \textbf{Efficiency}: average size of the prediction sets:
    \[
    \frac{1}{n_{\text{test}}} \sum_{i=1}^{n_{\text{test}}} \texttt{card}(\mathcal{C}(x_i))
    \]
    
    \item \textbf{Informativeness} proportion of prediction sets of size 1:
    \[
    \frac{1}{n_{\text{test}}} \sum_{i=1}^{n_{\text{test}}} \mathbbm{1}\{ \texttt{card}(\mathcal{C}(x_i)) = 1 \}
    \]
\end{itemize}
\end{definition}

Therefore, if a CP approach performs well, we expect (i) the coverage to be close to $1 - \alpha$, (ii) the efficiency to be closest to 1 from above as possible, and (iii) the informativeness to be closest to 1 from below as possible. For all settings (simulations and benchmark datasets), where possible, we replicate all experimental settings and parameters used in \cite{c1}.

\subsection{Simulated Data}
\label{sec:simulations}

We consider a low-dimensional binary classification problem with eight features. The feature matrix for class 1 is sampled from a Gaussian $\mathcal{N}(\bm{\mu}_1, \bm{\Sigma}_1)$, with $\bm{\mu}_1 = [0.8, \hdots, 0.8]^T \in \mathbb{R}^8$ and $\bm{\Sigma}_1 = \mathbb{I} \cdot 7 \in \mathbb{R}^{8 \times 8}$, while class 2 follows $\mathcal{N}(\bm{\mu}_2, \bm{\Sigma}_2)$ with $\bm{\mu}_2 = [-1, \hdots, -1]$ and $\bm{\Sigma}_2 = \mathbb{I} \cdot 8$, where $\mathbb{I}$ denotes the identity matrix. This setting delivers a reasonable overlap between the two classes and implies a slightly non-linear decision boundary due to the small difference in the covariance matrices. The generated data is class-balanced, with $ 60\% $ used for the training set, $24\%$ for the calibration set, and $16\%$ reserved for evaluating CP performance. As for the models, considering the slight non-linearity of the decision boundary (and assuming, however, that in reality we would not know the underlying data-generating process), we consider the two classifiers mentioned earlier: Na\"ive Bayes (NB) and Random Forest (RF) (we will see applications with deep neural networks in Section \ref{sec:real_data}). In particular, in this section we will focus on CP performance when employing the DP versions of these models (i.e. $\hat{f}^{DP}$), whereas the results using the non-DP models $\hat{f}$ can be found in the appendix. This being said, we evaluate coverage, efficiency, and informativeness with a fixed privacy budget $\epsilon_f = 2$ when privatizing the models and repeat experimental runs $H = 1000$ times in each setting (unless otherwise specified). Moreover, when varying one of the parameters to study the sensitivity of the methods, the other parameters are fixed at: $\epsilon_{CP} = 1$ (privacy budget for CP); $n = 10,000$ (total sample size to be split); $\alpha = 0.1$ (significance level for coverage). For clarity of presentation, here we only discuss the results for the NB classifier, while we provide the same results for the RF classifier in Appendices \ref{app:NonDPrandom_forest} and \ref{app:DPrandom_forest}.

%\textcolor{red}{HERE YOU TALK ABOUT $\epsilon_{CONFORMAL}$ AFTER HAVING TALKED ABOUT $\rho$-zCDP IN ALL THE WORK: IF YOU KEEP IT THIS WAY WE NEED TO STATE WHAT THIS CORRESPONDS TO IN TERMS OF zCDP! IS WHAT I STATED EARLIER CORRECT? IF SO THIS IS SOLVED.} \textcolor{blue}{They are same since we mentioned the equivalence of $\epsilon$-DP and $\rho$-zCDP. We just stick to one notation in that case.} \textcolor{red}{BUT WE HAVE TO SAY WHAT $\epsilon = 1$ CORRESPONDS TO.}

As mentioned earlier, in the following sections, we study the effects of the privacy budget $\epsilon_{CP}$ for CP in the \textsc{ExponQ} and P-COQS approaches, as well as of the total sample size (which is then split into training, calibration, and test) and significance level $\alpha$. However, it must be stressed that all results are obviously limited to the above data-generating mechanism and therefore are to be interpreted within this setting.

\subsubsection{Effect of CP Privacy Budget ($\epsilon_{CP}$)}

We examine the sensitivity of both \textsc{P-COQS} and \textsc{ExponQ} to variations in the CP privacy parameter $\epsilon_{CP}$, using private NB model while maintaining a fixed sample size of $n = 10,000$. For the DP model (achieving 75\% accuracy), the results in Table~\ref{PrivconformalNB} indicate that \textsc{P-COQS}  shows significantly lower sensitivity to $\epsilon_{CP}$ compared to \textsc{ExponQ}, while comparing favorably to it over the different evaluation metrics across the considered privacy budgets. Similar behavior is observed for the private RF model (see Table~\ref{PrivconformalRF} in Appendix~\ref{prandomforest}) and also when using the non-DP variants of the NB and RF models reported in Table~\ref{npPrivconformalNB} and Table~\ref{npPrivconformalrf} respectively, where as expected, efficiency and informativeness results are slightly better for both approaches.

\begin{table}[h!]
\centering
\resizebox{\textwidth}{!}{%
\begin{tabular}{lcc|cc|cc}
\toprule
\textbf{$\epsilon_{CP}$} 
& \multicolumn{2}{c|}{\textbf{Coverage}} 
& \multicolumn{2}{c|}{\textbf{Efficiency}} 
& \multicolumn{2}{c}{\textbf{Informativeness}} \\
\cmidrule(lr){2-3} \cmidrule(lr){4-5} \cmidrule(lr){6-7}
& \textbf{ExponQ} & \textbf{P-COQS} 
& \textbf{ExponQ} & \textbf{P-COQS} 
& \textbf{ExponQ} & \textbf{P-COQS} \\
\midrule
0.1  & 1.0000 (0.0002) & 0.9000 (0.0107) & 1.9995 (0.0030) & 1.3991 (0.0350) & 0.0005 (0.0030) & 0.6009 (0.0350) \\
0.5  & 0.9402 (0.0141) & 0.8999 (0.0100) & 1.5692 (0.0694) & 1.3985 (0.0326) & 0.4308 (0.0694) & 0.6015 (0.0326) \\
1    & 0.9217 (0.0116) & 0.8999 (0.0098) & 1.4841 (0.0457) & 1.3985 (0.0322) & 0.5159 (0.0457) & 0.6015 (0.0322) \\
1.5  & 0.9148 (0.0104) & 0.9000 (0.0098) & 1.4553 (0.0385) & 1.3986 (0.0319) & 0.5447 (0.0385) & 0.6014 (0.0319) \\
3    & 0.9076 (0.0098) & 0.9000 (0.0097) & 1.4272 (0.0337) & 1.3987 (0.0318) & 0.5728 (0.0337) & 0.6013 (0.0318) \\
5    & 0.9046 (0.0097) & 0.8999 (0.0097) & 1.4160 (0.0325) & 1.3986 (0.0320) & 0.5840 (0.0325) & 0.6014 (0.0320) \\
10   & 0.9024 (0.0096) & 0.8999 (0.0097) & 1.4073 (0.0321) & 1.3986 (0.0320) & 0.5927 (0.0321) & 0.6014 (0.0320) \\
\bottomrule
\end{tabular}%
}
\caption{Average effect of CP privacy budget $\epsilon_{CP}$ on conformal prediction with DP Na\"ive Bayes model. The numbers are metric averages over $1000$ runs (per method and privacy budget) and in parentheses is the corresponding variance of the metrics.}
\label{PrivconformalNB}
\end{table}

% \subsubsection{Effect of Model Privacy Budget ($\epsilon_{f}$)}
% \label{sec:epsilon_f_effect}

% We analyze the sensitivity of \textsc{P-COQS} and \textsc{ExponQ} to variations in the model privacy parameter $\epsilon_{f}$, with fixed sample size ($n=10,000$) and CP privacy parameter ($\epsilon_{CP}=1$). The results in Table~\ref{modelprivacyNB} demonstrate that \textsc{P-COQS} consistently outperforms \textsc{ExponQ} across all evaluated metrics for every $\epsilon_{f}$ value except $\epsilon_{f}=1$. This exceptional case requires further investigation to understand the underlying reasons for the performance parity.

% Notably, when examining the Random Forest model results in Table~\ref{modelprivacyRF}, we observe that \textsc{P-COQS} maintains its superior performance across all values of $\epsilon_{f}$. 

%An aspect that must be underlined and supports the result of Theorem \ref{thm:cov_guarantee} is the fact that, generally, for smaller privacy budgets \textsc{P-COQS} metrics tend to have larger variability compared to \textsc{ExponQ}, which then reduces towards the variability of the \textsc{ExponQ}  metrics as the privacy budget increases (similar conclusions hold for the DP RF model; see Appendix~\ref{prandomforest})

\subsubsection{Sample Size Effect}
We investigate the impact of sample size on the performance of the CP methods considered here. For each sample size, the sample is split as mentioned at the start of this section (i.e. $60\%$ training, $24\%$ calibration and $16\%$ test). We recall that we discuss the results with the DP version of NB (see Appendix \ref{samplesizeNB} for the result of the non-DP model $\hat{f}$) and also recall that, as in other experiments, the privacy budget for the models is fixed at $\epsilon_{f} = 2$. The results in Table~\ref{samplesiznb} show that in general P-COQS targets the correct $1 - \alpha$ coverage except in smaller sample settings (that is, $n \leq 200$) and overall performs well compared to \textsc{ExponQ}. More in detail, as discussed in the previous section, it can be seen how the sample size has an important effect on the P-COQS precision (measured in terms of metric variability) where, aside from visibly under-covering with respect to the $1 - \alpha$ target-level, smaller sample sizes show lower precision of P-COQS compared to \textsc{ExponQ} (similar conclusions hold for RF in Appendices ~\ref{samplesizeRF} and \ref{samplesizeNDP_RF}). This confirms the observations and conclusions made from Theorem \ref{thm:cov_guarantee}. Nevertheless, it can be observed that the P-COQS under-covering appears to be solved already starting from $n = 200$ and does not appear to be an issue when using a non-DP model (see Table \ref{samplesizeNB} in the Appendix). In addition, the variability of the P-COQS metrics tend to converge towards that of \textsc{ExponQ} (as also seen in the other simulation settings of this section).

\begin{table}[h!]
\centering
\resizebox{\textwidth}{!}{%
\begin{tabular}{lcc|cc|cc|c}
\toprule
\textbf{$n$} 
& \multicolumn{2}{c|}{\textbf{Coverage}} 
& \multicolumn{2}{c|}{\textbf{Efficiency}} 
& \multicolumn{2}{c|}{\textbf{Informativeness}} 
& \textbf{Model Accuracy} \\
\cmidrule(lr){2-3} \cmidrule(lr){4-5} \cmidrule(lr){6-7}
& \textbf{ExponQ} & \textbf{P-COQS} 
& \textbf{ExponQ} & \textbf{P-COQS} 
& \textbf{ExponQ} & \textbf{P-COQS} 
& \\
\midrule
100    & 0.9791 (0.0447) & 0.8529 (0.1472) & 1.9266 (0.1180) & 1.5958 (0.3184) & 0.0734 (0.1180) & 0.4042 (0.3184) & 0.5867 (0.1248) \\
200    & 0.9954 (0.0148) & 0.9295 (0.0731) & 1.9896 (0.0287) & 1.8077 (0.1663) & 0.0104 (0.0287) & 0.1923 (0.1663) & 0.5543 (0.1039) \\
500    & 0.9969 (0.0074) & 0.9043 (0.0474) & 1.9830 (0.0274) & 1.6153 (0.1206) & 0.0170 (0.0274) & 0.3847 (0.1206) & 0.6289 (0.0612) \\
1000   & 0.9994 (0.0034) & 0.9296 (0.0544) & 1.9963 (0.0147) & 1.7446 (0.1858) & 0.0037 (0.0147) & 0.2554 (0.1858) & 0.6266 (0.0478) \\
2000   & 0.9761 (0.0130) & 0.9014 (0.0220) & 1.8469 (0.0577) & 1.5316 (0.0644) & 0.1531 (0.0577) & 0.4684 (0.0644) & 0.6972 (0.0268) \\
6000   & 0.9340 (0.0178) & 0.9185 (0.0414) & 1.6877 (0.0859) & 1.6216 (0.1890) & 0.3123 (0.0859) & 0.3784 (0.1890) & 0.6949 (0.0227) \\
10000  & 0.9217 (0.0116) & 0.8999 (0.0098) & 1.4841 (0.0457) & 1.3985 (0.0322) & 0.5159 (0.0457) & 0.6015 (0.0322) & 0.7479 (0.0125) \\
\bottomrule
\end{tabular}%
}
\caption{Average effect of sample size on private conformal prediction with DP Na\"ive Bayes model. The numbers are metric averages over $1000$ runs (per method and sample size) and in parentheses is the corresponding variance of the metrics.}
\label{samplesiznb}
\end{table}

\subsubsection{Effect of Significance Level ($\alpha$)}

We examine the performance across varying $\alpha$ values (fixing $n = 10,000$, $\epsilon_{f} = 2$, and $\epsilon_{CP} = 1$). As seen in Table~\ref{alphaeffectNB}, P-COQS appears to better target the required coverage level on average, but has a higher variability than \textsc{ExponQ} in this metric with a small $\alpha = 0.01$ level. This variability then becomes in line with (or smaller than) that of \textsc{ExponQ} with larger values of $\alpha$. With respect to the other metrics, the results indicate that P-COQS generally performs better than \textsc{ExponQ} with the variability of the performance in line with the latter (similar conclusions holds for the DP RF model in Appendix~\ref{alphaeffectRF} as well as when using non-DP NB and RF models in Appendices~\ref{alpha_NP_NB} and \ref{alpha_NP_RF} respectively).

\begin{table}[h!]
\centering
\resizebox{\textwidth}{!}{%
\begin{tabular}{lcc|cc|cc}
\toprule
\textbf{$\alpha$} 
& \multicolumn{2}{c|}{\textbf{Coverage}} 
& \multicolumn{2}{c|}{\textbf{Efficiency}} 
& \multicolumn{2}{c}{\textbf{Informativeness}} \\
\cmidrule(lr){2-3} \cmidrule(lr){4-5} \cmidrule(lr){6-7}
& \textbf{ExponQ} & \textbf{P-COQS} 
& \textbf{ExponQ} & \textbf{P-COQS} 
& \textbf{ExponQ} & \textbf{P-COQS} \\
\midrule
0.01  & 1.0000 (0.0002) & 0.9988 (0.0045) & 1.9995 (0.0030) & 1.9857 (0.0511) & 0.0005 (0.0030) & 0.0143 (0.0511) \\
0.05  & 0.9732 (0.0088) & 0.9500 (0.0076) & 1.7579 (0.0604) & 1.6152 (0.0373) & 0.2421 (0.0604) & 0.3848 (0.0373) \\
0.10  & 0.9217 (0.0116) & 0.8999 (0.0098) & 1.4841 (0.0457) & 1.3985 (0.0322) & 0.5159 (0.0457) & 0.6015 (0.0322) \\
\bottomrule
\end{tabular}%
}
\caption{Average effect of $\alpha$ on private conformal prediction using DP Na\"ive Bayes model. The numbers are metric averages over $1000$ runs (per method and significance level) and in parentheses is the corresponding variance of the metrics.}
\label{alphaeffectNB}
\end{table}

\subsubsection{Computational Efficiency}

%\textcolor{red}{DO THESE RUNTIMES INCLUDE MODEL TRAINING? OR IS IT ONLY CP? IT SHOULD BE THE LATTER.} \textcolor{blue}{oops! It includes model training. I have rerun the computational time; it is now for CP only}

We compare the runtime of the considered methods under two settings which considered the non-DP and DP versions of each model (i.e. $\hat{f}$ and $\hat{f}^{DP}$ respectively). Confirming the model privacy budget to be $\epsilon_{f} = 2$, conformal privacy budget $\epsilon_{CP} = 1$, total sample size  $n = 10,000$ and significance level $\alpha = 0.1$, we record the average runtime over 1000 trials. Table~\ref{averagetimeNB} highlights the considerable computational gain that P-COQS delivers with respect to \textsc{ExponQ} while remaining comparable (or often better) across all performance metrics considered above (similar conclusions holds for the RF model in Appendix~\ref{AveragetimeRF}). It must be underlined that, as mentioned in Section \ref{sec:rel_work}, the increased computational time for \textsc{ExponQ} is mainly due to the optimization subroutine needed to choose the optimal number of bins and the quantile-inflation hyperparameter to achieve the smallest possible prediction set.

%\textcolor{red}{I DON'T UNDERSTAND THIS: YOU RUN THESE 1000 TIMES AND THEN ANOTHER THREE? SO 3000 TIMES?} \textcolor{blue}{Yes. Average of the three run times; each for 1000 iterations. I should have done it as done in the Fima paper; recording the time for each iteration and compute the average time as the average of the 1000 computational times. I will revisit this.}\textcolor{green}{THIS IS NOW FIXED. I HAVE UPDATED THE TABLE}

\begin{table}[h!]
\centering
\resizebox{\textwidth}{!}{%
\begin{tabular}{lcccccc}
    \toprule
    & Coverage & Efficiency & Informativeness & Model Accuracy & Ave. time (secs) \\
    \midrule
    %Npriv\_model and Npriv\_conform & 0.9004 (0.0098) & 1.1783 (0.0196) & 0.8217 (0.0196) & 0.8253 (0.0092) & - \\
    $\hat{f}$ and \textsc{ExponQ} & 0.9227 (0.0116) & 1.2509 (0.0360) & 0.7491 (0.0360) & 0.8253 (0.0092) & 0.7562 (0.0111) \\
    $\hat{f}$ and P-COQS & 0.9006 (0.0010) & 1.1788 (0.0201) & 0.8212 (0.0201) & 0.8253 (0.0092) & 0.0072 (0.0002)  \\
    %Priv\_model and Npriv\_conform & 0.8998 (0.0096) & 1.3981 (0.0318) & 0.6019 (0.0318) & 0.7479 (0.0125) & - \\
    $\hat{f}^{DP}$ and \textsc{ExponQ} & 0.9217 (0.0116) & 1.4841 (0.0457) & 0.5159 (0.0457) & 0.7479 (0.0125) & 0.7534 (0.0152) \\
    $\hat{f}^{DP}$ and P-COQS & 0.8999 (0.0098) & 1.3985 (0.0322) & 0.6015 (0.0322) & 0.7479 (0.0125) & 0.0072 (0.0003) \\
    \bottomrule
\end{tabular}%
}
\caption{Average computational time of P-COQS and \textsc{ExponQ} combined with non-DP and DP versions of the NB classifier.}
\label{averagetimeNB}
\end{table}

\subsection{Datasets}
\label{sec:real_data}

% \textcolor{red}{YOU WILL HAVE TO CHANGE THE NAME OF OUR METHOD IN THE PLOTS (FROM PRIVBINSQ TO P-COQS), AS WELL AS THE NOTATION $\epsilon_{CP}$ INSTEAD OF $\epsilon_{Conformal}$.}

In this section, we take the same settings as in \cite{c1} and study the performance of P-COQS when applied to three benchmark datasets: CIFAR-10 \citep{krizhevsky2009}, ImageNet \citep{deng2009} and CoronaHack \citep{perez2020}. As in the simulation experiment, we compare the performance of P-COQS with that of \textsc{ExponQ} proposed by \citet{c1} and follow the same experimental setup as in the latter. Unless otherwise stated, the significance level is fixed at $\alpha = 0.1$, CP privacy budget at $\epsilon_{CP} = 1$ and, when used, the DP models are trained to achieve $(\epsilon, \delta)$-DP using the \texttt{Opacus} library with $(\epsilon = 8$ and $\delta = 10^{-5})$.

%\textcolor{red}{WHAT DO YOU MEAN BY CLOSELY? IS IT NOT EXACTLY THE SAME?} \textcolor{blue}{There is a slight change in the architecture of the DNN such as increasing the number of layers. But the privacy parameters, sample sizes, number of iterations, and conformal prediction settings remain the same as in their paper. This was done in order to achieve a reasonable test accuracy since some packages in the Opacus library they used to train their DP models, are no longer available, so I could not use their training code and settings}.

\subsubsection{CIFAR-10 Benchmark Analysis}

We first evaluate P-COQS on the CIFAR-10 dataset, comparing its performance to \textsc{ExponQ} under different privacy settings. Following \citet{c1}, we consider two scenarios for model training: (i) one where a non-DP model is used ($\hat{f}$) and (ii) one where a DP model is used ($\hat{f}^{DP}$). Both DP and non-DP models share the same convolutional neural network architecture. The DP model achieved an accuracy of $60\%$ compared to $77\%$ for the non-DP model. The evaluation is based on 1000 random splits of the CIFAR-10 validation set, each of size $n = 5000$. Figure~\ref{fig:results_cf10} displays the empirical coverage and prediction set sizes across these settings. Our findings suggest that P-COQS generally produces smaller prediction sets (Figure~\ref{fig:sizes_cf_10}), while achieving empirical coverage levels that are comparable to those of the standard (non-private) CP baseline (Figure~\ref{fig:coverage_cf_10}), except for the case where a DP model is used where P-COQS appears to slightly undercover with respect to the standard CP baseline. On the other hand, compared to P-COQS, by inflating the quantile \textsc{ExponQ} guarantees the coverage level but tends to have larger prediction sets for this reason. In the latter case, for all methods it can be seen that prediction sets are larger when using a DP model (as expected).

% \begin{figure}[htbp]
%   \centering
%   \begin{subfigure}[b]{0.49\textwidth}
%     \centering
%     \includegraphics[width=\textwidth]{CIFAR_10_coverage_boxplot.pdf}
%     \caption{
%       Empirical coverage probability ($\alpha=0.1$). 
%       %The dashed line shows the target coverage (90\%).
%     }
%     \label{fig:coverage_cf_10}
%   \end{subfigure}
%   \hfill
%   \begin{subfigure}[b]{0.49\textwidth}
%     \centering
%     \includegraphics[width=\textwidth]{CIFAR_10_size_boxplot.pdf}
%     \caption{
%       Prediction set sizes ($\alpha=0.1$).
%     }
%     \label{fig:sizes_cf_10}
%   \end{subfigure}
%   \caption{
%     Performance of private conformal prediction with CIFAR-10 dataset. 
%     (a) Coverage probabilities where the dashed line shows the target coverage (90\%) and (b) prediction set sizes at $\alpha=0.1$.
%   }
%   \label{fig:results_cf10}
% \end{figure}

\begin{figure}[ht]
    \centering
    \begin{minipage}[t]{0.49\textwidth}
        \centering
        \includegraphics[height=4.5cm]{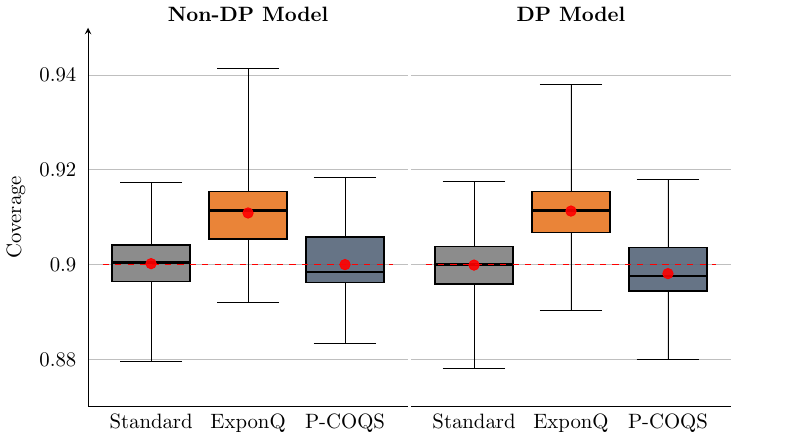}
        \subcaption{Coverage}
        \label{fig:coverage_cf_10}
    \end{minipage}
    \hfill
    \begin{minipage}[t]{0.49\textwidth}
        \centering
        \includegraphics[height=4.5cm]{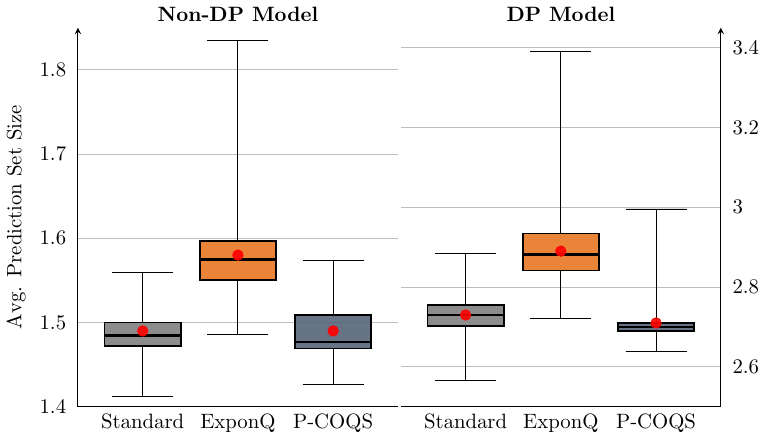}
        \subcaption{Avg. Prediction Set Size}
        \label{fig:sizes_cf_10}
    \end{minipage}
    \caption{Results of Standard, \textsc{ExponQ} and P-COQS methods applied to the CIFAR-10 dataset. (a) Boxplots with coverage distributions of the three methods under the non-DP model (left) and DP model (right) scenarios; (b) Boxplots with distributions of average set sizes of the three methods under the non-DP model (left) and DP model (right) scenarios.}
    \label{fig:results_cf10}
\end{figure}

\subsubsection{Large-Scale Evaluation on ImageNet}

We apply both considered private CP methods to the ImageNet dataset using a pre-trained (non-DP) ResNet-152 model. We evaluated the performance of P-COQS and \textsc{ExponQ} under varying values of the conformal privacy budget $\epsilon_{CP}$. The evaluation uses $n = 30,000$ samples for calibration and $20,000$ samples for validation. For each budget value $\epsilon_{CP}$, performance metrics are computed over 100 random splits of the ImageNet validation set, following the procedure in \citet{c1}. As shown in Figure~\ref{fig:results_imagenet}, P-COQS appears to slightly under-cover but this does not appear to be significant based on the 95\% confidence intervals (red vertical whiskers in the plots) and remains stable in performance across the different values of conformal privacy budget, whereas \textsc{ExponQ} appears to significantly over-cover for the small budgets (see Figure \ref{fig:coverage_imagenet}). With this in mind, P-COQS also appears to be stable with respect to prediction set sizes while \textsc{ExponQ} approaches these sizes only as $\epsilon_{CP}$ increases (as highlighted also in the simulation results in Section \ref{sec:simulations}).

% \begin{figure}[htbp]
%   %\centering
%   \begin{subfigure}[b]{0.49\textwidth}
%     %\centering
%     \includegraphics[width=\textwidth]{ImangeNet_coverage_boxplots_alpha_0.1_n_30000.pdf}
%     \caption{
%       Empirical coverage probability ($\alpha=0.1$). 
%       %The dashed line shows the target coverage (90\%).
%     }
%     \label{fig:coverage_imagenet}
%   \end{subfigure}
%   \hfill
%   \begin{subfigure}[b]{0.49\textwidth}
%     %\centering
%     \includegraphics[width=\textwidth]{ImangeNet_size_boxplots_alpha_0.1_n_30000.pdf}
%     \caption{
%       Prediction set sizes ($\alpha=0.1$).
%     }
%     \label{fig:sizes_imagenet}
%   \end{subfigure}
%   \caption{
%     Performance of private conformal prediction. 
%     (a) Coverage probabilities and (b) prediction set sizes at $\alpha=0.1$.
%   }
%   \label{fig:results_imagenet}
% \end{figure}

\begin{figure}[ht]
    \centering
    \begin{minipage}[t]{0.49\textwidth}
        \centering
        \includegraphics[height=5cm]{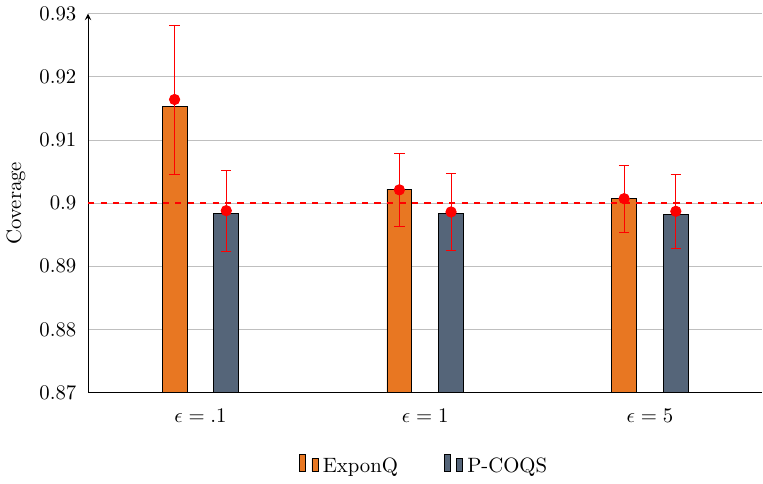}
        \subcaption{Coverage}
        \label{fig:coverage_imagenet}
    \end{minipage}
    %\hfill
    \begin{minipage}[t]{0.49\textwidth}
        \centering
        \includegraphics[height=5cm]{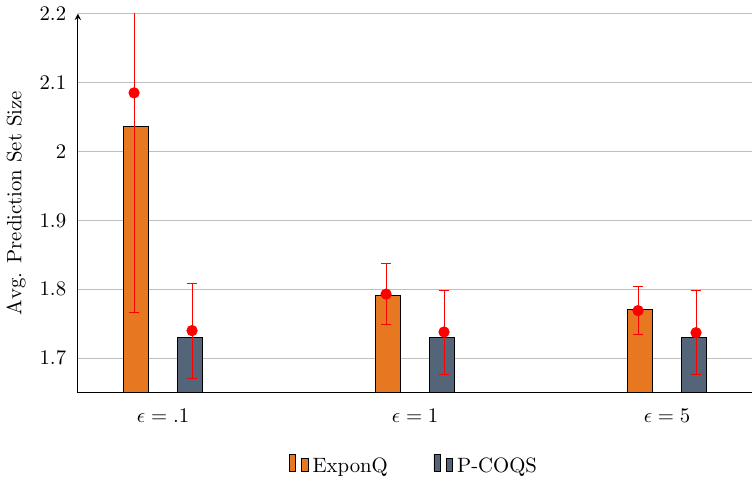}
        \subcaption{Avg. Prediction Set Size}
        \label{fig:sizes_imagenet}
    \end{minipage}
    \caption{Results of \textsc{ExponQ} and P-COQS methods applied to the ImageNet dataset. (a) Median coverage (bar height), mean coverage (red dots) and 95\% confidence intervals for coverage (red vertical whiskers) for each method under different privacy budgets $\epsilon$ (red-dashed horizontal line is the target coverage); (b) Average prediction set sizes (average set sizes in 100 splits): median (bar height), mean (red dots) and 95\% confidence intervals (red vertical whiskers) for each method under different privacy budgets $\epsilon$.}
    \label{fig:results_imagenet}
\end{figure}

\subsubsection{Medical Imaging Analysis with CoronaHack Dataset}

Finally, we evaluate both methods on the CoronaHack dataset, a publicly available chest X-ray dataset comprising 5908 images categorized into three classes: normal, viral pneumonia (primarily COVID-19), and bacterial pneumonia. We fine-tuned the final layer of a ResNet-18 model over 14 epochs using 4408 training samples, under both DP and non-DP settings. The DP model achieved $64\%$ accuracy compared to $70\%$ for the non-DP model. CP was calibrated on 1000 samples and validated on 500 samples, across 1000 random splits of the data. Figure~\ref{fig:results_ch} presents the results for different configurations of the DP and non-DP models. While \textsc{ExponQ} tends to produce overly conservative prediction sets, P-COQS appears to better target the nominal coverage level of $0.9$ (see Figure \ref{fig:coverage_ch}). Moreover, looking at the prediction set sizes, the distribution of the P-COQS appears to deliver roughly similar set sizes to the non-private CP approach since the proportions of sets of specific sizes (1, 2 or 3 on the x-axis) are very close to each other (see Figure \ref{fig:sizes_ch}). Also in this case \textsc{ExponQ} produces larger sets (i.e. higher proportions for sets of size 2 or 3), where again the sizes tend to increase overall when using the DP model.

\begin{figure}[ht]
    \centering
    \begin{minipage}[t]{0.43\textwidth}
        \raggedright
        \includegraphics[height=4.8cm]{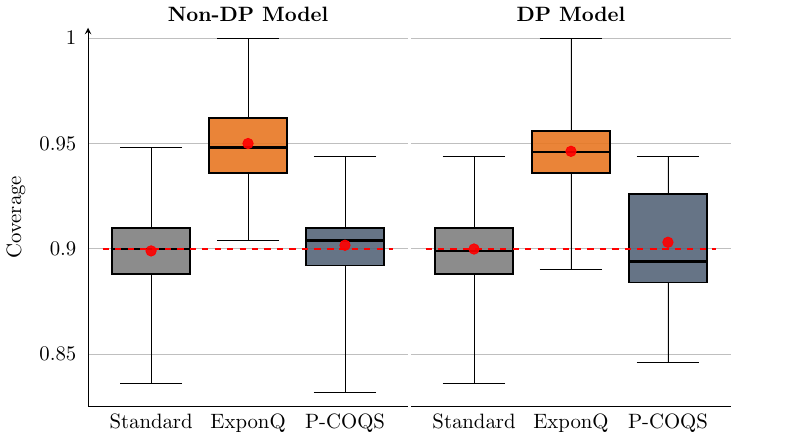}
        \subcaption{Coverage}
        \label{fig:coverage_ch}
    \end{minipage}
    \hfill
    \begin{minipage}[t]{0.53\textwidth}
        \centering
        \includegraphics[height=4.7cm]{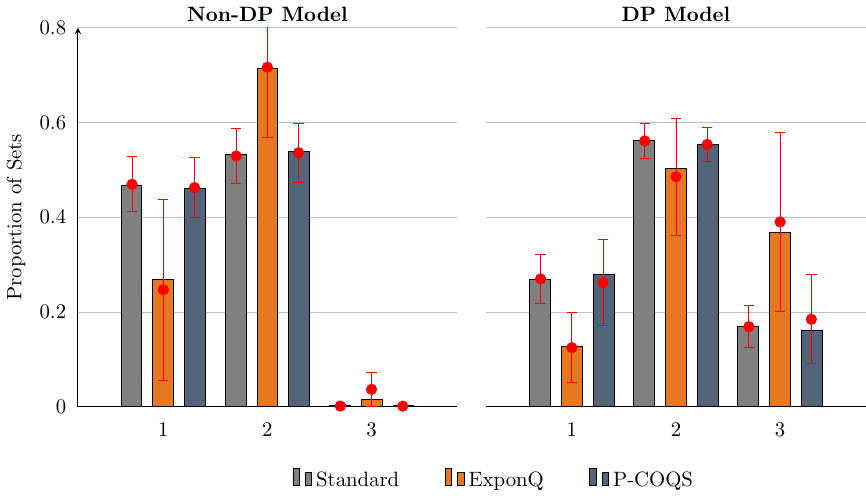}
        \subcaption{Distribution of Prediction Set Sizes}
        \label{fig:sizes_ch}
    \end{minipage}
    \caption{Results of Standard, \textsc{ExponQ} and P-COQS methods applied to the CoronaHack dataset. (a) Boxplots with coverage distributions of the three methods under the non-DP model (left) and DP model (right) scenarios; (b) Distribution of prediction set sizes where the y-axis represents the empirical proportion of sets of size 1, 2 or 3 (represented on the x-axis) for each method: median proportion (bar height), mean proportion (red dots) and 95\% confidence intervals (red vertical whiskers) for each method under the non-DP model (left) and DP model (right) scenarios.}
    \label{fig:results_ch}
\end{figure}

% \begin{figure}[htbp]
%   %\centering
%   \begin{subfigure}[b]{0.49\textwidth}
%     %\centering
%     \includegraphics[width=\textwidth]{coronahack_coverage.pdf}
%     \caption{
%       Empirical coverage probability ($\alpha=0.1$). 
%       %The dashed line shows the target coverage (90\%).
%     }
%     \label{fig:coverage_ch}
%   \end{subfigure}
%   \hfill
%   \begin{subfigure}[b]{0.49\textwidth}
%     \centering
%     \includegraphics[width=\textwidth]{CH_experiment_size_histograms.pdf}
%     \caption{
%       Prediction set sizes ($\alpha=0.1$).
%     }
%     \label{fig:sizes_ch}
%   \end{subfigure}
%   \caption{
%     Performance of private conformal prediction. 
%     (a) Coverage probabilities and (b) prediction set sizes at $\alpha=0.1$.
%   }
%   \label{fig:results_ch}
% \end{figure}

\section{Discussion}

The proposed P-COQS provides a computationally efficient DP alternative to existing privacy-preserving approaches. It does so by adapting an existing DP quantile method and trading off some theoretical guarantees on the lower bound for the coverage probability, while however guaranteeing quantifiable lower and upper error bounds for this probability. In this way, it is possible to determine (tight) error bounds which can then be used to assess the true minimum coverage level with high probability. Despite this theoretical trade-off, which however remains small (if not negligible) in common settings, the P-COQS shows an empirical performance over different metrics that is in line with standard non-private CP, indicating that it targets the correct coverage level in all settings while delivering more precise CP sets compared to the existing DP alternative (that we denoted as \textsc{ExponQ}). The theoretical trade-off of the P-COQS can be observed mainly in settings with small sample sizes and/or privacy budgets in which the variability of P-COQS metrics is indeed larger compared to \textsc{ExponQ}. However, the latter approach also makes a trade-off by generally over-covering with respect to the required coverage level in smaller samples and for small privacy budgets. As a consequence, the latter tends to deliver less efficient and informative prediction sets in these settings and appear to only approach those of P-COQS for larger privacy budgets and/or sample sizes. Overall, considering these different trade-offs and without claiming the existence of a better approach, the proposed P-COQS seems to provide a valid privacy-preserving (and computationally efficient) CP approach based also on the generally good performance compared to \textsc{ExponQ} in the different experimental settings considered.

% In this paper, we proposed a binary search-based algorithm for computing a private quantile in the calibration step of conformal prediction for classification tasks. This method is simple, computationally efficient, and avoids the discretization and binning of calibration scores as done in \cite{c1}, thus avoiding discretization errors. 

% We carried out extensive simulation studies investigating the impact of privacy noise, sample size, and significance level on the performance of the proposed method relative to existing alternatives. In addition, we  empirically evaluated our method on several benchmark datasets, including CIFAR-10, ImageNet, and CoronaHack. 

% The results demonstrate that the proposed method consistently outperforms the method proposed in \citet{c1} in terms coverage, efficiency, and informativeness. Specifically, it produces smaller conformal prediction sets that still meet the desired coverage guarantee while ensuring strong differential privacy protection.

% Acknowledgements and Disclosure of Funding should go at the end, before appendices and references

\acks{
In preparing this work, ChatGPT 4o was used to outline and proofread. The authors subsequently carefully reviewed and edited the content to ensure accuracy and coherence. The authors take full responsibility for the integrity of this work. This work has been funded by NSF-SES 2150615.}


\newpage

\bibliography{references.bib}

\begin{thebibliography}{37}
\providecommand{\natexlab}[1]{#1}
\providecommand{\url}[1]{\texttt{#1}}
\expandafter\ifx\csname urlstyle\endcsname\relax
  \providecommand{\doi}[1]{doi: #1}\else
  \providecommand{\doi}{doi: \begingroup \urlstyle{rm}\Url}\fi

\bibitem[Angelopoulos and Bates(2021)]{c39}
A.~N. Angelopoulos and S.~Bates.
\newblock A gentle introduction to conformal prediction and distribution-free uncertainty quantification.
\newblock \emph{arXiv preprint arXiv:2107.07511}, 2021.

\bibitem[Angelopoulos et~al.(2022)Angelopoulos, Bates, Zrnic, and Jordan]{c1}
A.~N. Angelopoulos, S.~Bates, T.~Zrnic, and M.~I. Jordan.
\newblock Private prediction sets.
\newblock \emph{Harvard Data Science Review}, 4\penalty0 (2), 2022.

\bibitem[Astigarraga et~al.(2025)Astigarraga, S{\'a}nchez-Ruiz, and Colmenarejo]{astigarraga2025conformal}
M.~Astigarraga, A.~S{\'a}nchez-Ruiz, and G.~Colmenarejo.
\newblock Conformal prediction-based machine learning in cheminformatics: Current applications and new challenges.
\newblock \emph{Artificial Intelligence in the Life Sciences}, page 100127, 2025.

\bibitem[Blundell et~al.(2015)Blundell, Cornebise, Kavukcuoglu, and Wierstra]{blundell2015weight}
C.~Blundell, J.~Cornebise, K.~Kavukcuoglu, and D.~Wierstra.
\newblock Weight uncertainty in neural networks.
\newblock In \emph{International Conference on Machine Learning}, pages 1613--1622. PMLR, 2015.

\bibitem[Bun and Steinke(2016)]{bun2016}
M.~Bun and T.~Steinke.
\newblock Concentrated differential privacy: Simplifications, extensions, and lower bounds.
\newblock In \emph{Theory of cryptography conference}, pages 635--658. Springer, 2016.

\bibitem[Carlini et~al.(2021)Carlini, Tramer, Wallace, Jagielski, Herbert-Voss, Lee, Roberts, Brown, Song, Erlingsson, et~al.]{carlini2021extracting}
N.~Carlini, F.~Tramer, E.~Wallace, M.~Jagielski, A.~Herbert-Voss, K.~Lee, A.~Roberts, T.~Brown, D.~Song, U.~Erlingsson, et~al.
\newblock Extracting training data from large language models.
\newblock In \emph{30th USENIX security symposium (USENIX Security 21)}, pages 2633--2650, 2021.

\bibitem[Deng et~al.(2009)Deng, Dong, Socher, Li, Li, and Fei-Fei]{deng2009}
J.~Deng, W.~Dong, R.~Socher, L.-J. Li, K.~Li, and L.~Fei-Fei.
\newblock Imagenet: A large-scale hierarchical image database.
\newblock In \emph{2009 IEEE conference on computer vision and pattern recognition}, pages 248--255. Ieee, 2009.

\bibitem[Dong et~al.(2022)Dong, Roth, and Su]{dong2022gaussian}
J.~Dong, A.~Roth, and W.~J. Su.
\newblock Gaussian differential privacy.
\newblock \emph{Journal of the Royal Statistical Society: Series B (Statistical Methodology)}, 84\penalty0 (1):\penalty0 3--37, 2022.

\bibitem[Dwork(2006)]{dwork2006differential}
C.~Dwork.
\newblock Differential privacy.
\newblock In \emph{International colloquium on automata, languages, and programming}, pages 1--12. Springer, 2006.

\bibitem[Fredrikson et~al.(2015)Fredrikson, Jha, and Ristenpart]{fredrikson2015model}
M.~Fredrikson, S.~Jha, and T.~Ristenpart.
\newblock Model inversion attacks that exploit confidence information and basic countermeasures.
\newblock In \emph{Proceedings of the 22nd ACM SIGSAC conference on computer and communications security}, pages 1322--1333, 2015.

\bibitem[Gal and Ghahramani(2016)]{gal2016dropout}
Y.~Gal and Z.~Ghahramani.
\newblock Dropout as a bayesian approximation: Representing model uncertainty in deep learning.
\newblock In \emph{International Conference on Machine Learning}, pages 1050--1059. PMLR, 2016.

\bibitem[Gibbs and Candès(2021)]{gibbs2021adaptive}
I.~Gibbs and E.~Candès.
\newblock Adaptive conformal inference under distribution shift.
\newblock \emph{arXiv preprint arXiv:2106.00170}, 2021.

\bibitem[Huang et~al.(2021)Huang, Liang, and Yi]{c2}
Z.~Huang, Y.~Liang, and K.~Yi.
\newblock Instance-optimal mean estimation under differential privacy.
\newblock \emph{Advances in Neural Information Processing Systems}, 34:\penalty0 25993--26004, 2021.

\bibitem[Humbert et~al.(2023)Humbert, Le~Bars, Bellet, and Arlot]{humbert2023one}
P.~Humbert, B.~Le~Bars, A.~Bellet, and S.~Arlot.
\newblock One-shot federated conformal prediction.
\newblock In \emph{International Conference on Machine Learning}, pages 14153--14177. PMLR, 2023.

\bibitem[Jonkers et~al.(2024)Jonkers, Van~Wallendael, Duchateau, and Van~Hoecke]{jonkers2024wcps}
J.~Jonkers, G.~Van~Wallendael, L.~Duchateau, and S.~Van~Hoecke.
\newblock Conformal predictive systems under covariate shift.
\newblock In \emph{Proceedings of the Thirteenth Symposium on Conformal and Probabilistic Prediction with Applications}, pages 406--423. PMLR, 2024.

\bibitem[Kendall and Gal(2017)]{kendall2017uncertainties}
A.~Kendall and Y.~Gal.
\newblock What uncertainties do we need in bayesian deep learning for computer vision?
\newblock In \emph{Advances in Neural Information Processing Systems}, pages 5574--5584, 2017.

\bibitem[Krizhevsky et~al.(2009)Krizhevsky, Hinton, et~al.]{krizhevsky2009}
A.~Krizhevsky, G.~Hinton, et~al.
\newblock Learning multiple layers of features from tiny images.
\newblock 2009.

\bibitem[Lakshminarayanan et~al.(2017)Lakshminarayanan, Pritzel, and Blundell]{lakshminarayanan2017simple}
B.~Lakshminarayanan, A.~Pritzel, and C.~Blundell.
\newblock Simple and scalable predictive uncertainty estimation using deep ensembles.
\newblock In \emph{Advances in Neural Information Processing Systems}, pages 6402--6413, 2017.

\bibitem[Lei et~al.(2015)Lei, Rinaldo, and Wasserman]{c5}
J.~Lei, A.~Rinaldo, and L.~Wasserman.
\newblock A conformal prediction approach to explore functional data.
\newblock \emph{Annals of Mathematics and Artificial Intelligence}, 74:\penalty0 29--43, 2015.

\bibitem[Lei et~al.(2018)Lei, G’Sell, Rinaldo, Tibshirani, and Wasserman]{lei2018distribution}
J.~Lei, M.~G’Sell, A.~Rinaldo, R.~J. Tibshirani, and L.~Wasserman.
\newblock Distribution-free predictive inference for regression.
\newblock \emph{Journal of the American Statistical Association}, 113\penalty0 (523):\penalty0 1094--1111, 2018.

\bibitem[Li et~al.(2025)Li, Klamkin, Tanneau, Zandehshahvar, and Van~Hentenryck]{li2025conformal}
M.~Li, M.~Klamkin, M.~Tanneau, R.~Zandehshahvar, and P.~Van~Hentenryck.
\newblock Conformal prediction with upper and lower bound models.
\newblock \emph{arXiv preprint arXiv:2503.04071}, 2025.

\bibitem[McSherry and Talwar(2007)]{c40}
F.~McSherry and K.~Talwar.
\newblock Mechanism design via differential privacy.
\newblock In \emph{48th Annual IEEE Symposium on Foundations of Computer Science (FOCS'07)}, pages 94--103. IEEE, 2007.

\bibitem[Near and Abuah(2021)]{near2021programming}
J.~P. Near and C.~Abuah.
\newblock Programming differential privacy.
\newblock \emph{URL: https://uvm}, 2021.

\bibitem[Olsson et~al.(2022)Olsson, Kartasalo, Mulliqi, Capuccini, Ruusuvuori, Samaratunga, Delahunt, Lindskog, Janssen, Blilie, et~al.]{olsson2022estimating}
H.~Olsson, K.~Kartasalo, N.~Mulliqi, M.~Capuccini, P.~Ruusuvuori, H.~Samaratunga, B.~Delahunt, C.~Lindskog, E.~A. Janssen, A.~Blilie, et~al.
\newblock Estimating diagnostic uncertainty in artificial intelligence assisted pathology using conformal prediction.
\newblock \emph{Nature communications}, 13\penalty0 (1):\penalty0 7761, 2022.

\bibitem[Osband et~al.(2016)Osband, Blundell, Pritzel, and Van~Roy]{osband2016deep}
I.~Osband, C.~Blundell, A.~Pritzel, and B.~Van~Roy.
\newblock Deep exploration via bootstrapped dqn.
\newblock In \emph{Advances in Neural Information Processing Systems}, pages 4026--4034, 2016.

\bibitem[Papadopoulos et~al.(2002)Papadopoulos, Proedrou, Vovk, and Gammerman]{c4}
H.~Papadopoulos, K.~Proedrou, V.~Vovk, and A.~Gammerman.
\newblock Inductive confidence machines for regression.
\newblock In \emph{Machine learning: ECML 2002: 13th European conference on machine learning Helsinki, Finland, August 19--23, 2002 proceedings 13}, pages 345--356. Springer, 2002.

\bibitem[Penso et~al.(2025)Penso, Mahpud, Goldberger, and Sheffet]{penso2025privacy}
C.~Penso, B.~Mahpud, J.~Goldberger, and O.~Sheffet.
\newblock Privacy-preserving conformal prediction under local differential privacy.
\newblock \emph{arXiv preprint arXiv:2505.15721}, 2025.

\bibitem[P{\'e}rez et~al.(2020)P{\'e}rez, Contreras, de~Blas~P{\'e}rez, and Alvarez]{perez2020}
J.~C. P{\'e}rez, J.~M.~C. Contreras, C.~de~Blas~P{\'e}rez, and F.~L. Alvarez.
\newblock Databiology lab coronahack: Collection of public covid-19 data.
\newblock \emph{bioRxiv}, pages 2020--10, 2020.

\bibitem[Shokri et~al.(2017)Shokri, Stronati, Song, and Shmatikov]{shokri2017membership}
R.~Shokri, M.~Stronati, C.~Song, and V.~Shmatikov.
\newblock Membership inference attacks against machine learning models.
\newblock In \emph{2017 IEEE symposium on security and privacy (SP)}, pages 3--18. IEEE, 2017.

\bibitem[Tibshirani et~al.(2019)Tibshirani, Barber, Candès, and Ramdas]{tibshirani2019conformal}
R.~J. Tibshirani, R.~F. Barber, E.~J. Candès, and A.~Ramdas.
\newblock Conformal prediction under covariate shift.
\newblock \emph{arXiv preprint arXiv:1904.06019}, 2019.

\bibitem[Tuwani and Beam(2023)]{tuwani2023safe}
R.~Tuwani and A.~Beam.
\newblock Safe and reliable transport of prediction models to new healthcare settings without the need to collect new labeled data.
\newblock \emph{medRxiv}, 2023.

\bibitem[Uddin and Lofstrom(2023)]{uddin2023applications}
N.~Uddin and T.~Lofstrom.
\newblock Applications of conformal regression on real-world industrial use cases using crepes and mapie.
\newblock In \emph{Conformal and Probabilistic Prediction with Applications}, pages 147--165. PMLR, 2023.

\bibitem[Vovk(2002)]{c9}
V.~Vovk.
\newblock On-line confidence machines are well-calibrated.
\newblock In \emph{The 43rd Annual IEEE Symposium on Foundations of Computer Science, 2002. Proceedings.}, pages 187--196. IEEE, 2002.

\bibitem[Vovk(2015)]{c6}
V.~Vovk.
\newblock Cross-conformal predictors.
\newblock \emph{Annals of Mathematics and Artificial Intelligence}, 74:\penalty0 9--28, 2015.

\bibitem[Vovk and Petej(2012)]{c11}
V.~Vovk and I.~Petej.
\newblock Venn-abers predictors.
\newblock \emph{arXiv preprint arXiv:1211.0025}, 2012.

\bibitem[Vovk et~al.(2005)Vovk, Gammerman, and Shafer]{c3}
V.~Vovk, A.~Gammerman, and G.~Shafer.
\newblock \emph{Algorithmic learning in a random world}, volume~29.
\newblock Springer, 2005.

\bibitem[Vovk et~al.(2017)Vovk, Shen, Manokhin, and Xie]{c10}
V.~Vovk, J.~Shen, V.~Manokhin, and M.-g. Xie.
\newblock Nonparametric predictive distributions based on conformal prediction.
\newblock In \emph{Conformal and probabilistic prediction and applications}, pages 82--102. PMLR, 2017.

\end{thebibliography}

\appendix

\section{Simulation Results With Non-DP Random Forest Model}\label{app:NonDPrandom_forest}

\subsection{Effect of $\epsilon_{CP}$ with non-DP model}\label{nprandomforest}

The parameter settings are as in the non-DP NB model and the model's accuracy is $81\%$.

\begin{table}[H]
\centering
\resizebox{\textwidth}{!}{%
\begin{tabular}{lcc|cc|cc}
\toprule
\textbf{$\epsilon_{\text{conformal}}$} 
& \multicolumn{2}{c|}{\textbf{Coverage}} 
& \multicolumn{2}{c|}{\textbf{Efficiency}} 
& \multicolumn{2}{c}{\textbf{Informativeness}} \\
\cmidrule(lr){2-3} \cmidrule(lr){4-5} \cmidrule(lr){6-7}
& \textbf{ExponQ} & \textbf{P-COQS} 
& \textbf{ExponQ} & \textbf{P-COQS} 
& \textbf{ExponQ} & \textbf{P-COQS} \\
\midrule
0.1  & 0.9998 (0.0004) & 0.9024 (0.0104) & 1.9885 (0.0123) & 1.2222 (0.0243) & 0.0115 (0.0123) & 0.7778 (0.0243) \\
0.5  & 0.9464 (0.0188) & 0.9025 (0.0098) & 1.4094 (0.1213) & 1.2223 (0.0220) & 0.5906 (0.1213) & 0.7777 (0.0220) \\
1    & 0.9228 (0.0123) & 0.9026 (0.0098) & 1.2936 (0.0408) & 1.2224 (0.0218) & 0.7064 (0.0408) & 0.7776 (0.0218) \\
1.5  & 0.9154 (0.0107) & 0.9026 (0.0098) & 1.2660 (0.0294) & 1.2225 (0.0217) & 0.7340 (0.0294) & 0.7775 (0.0217) \\
3    & 0.9084 (0.0099) & 0.9026 (0.0098) & 1.2415 (0.0235) & 1.2225 (0.0220) & 0.7585 (0.0235) & 0.7775 (0.0220) \\
5    & 0.9058 (0.0099) & 0.9026 (0.0099) & 1.2327 (0.0224) & 1.2225 (0.0220) & 0.7673 (0.0224) & 0.7775 (0.0220) \\
10   & 0.9041 (0.0100) & 0.9026 (0.0099) & 1.2273 (0.0224) & 1.2226 (0.0219) & 0.7727 (0.0224) & 0.7774 (0.0219) \\
\bottomrule
\end{tabular}%
}
\caption{Average effect of CP privacy budget $\epsilon_{CP}$ on conformal prediction with non-DP Random Forest model. The numbers are metric averages over $1000$ runs (per method and privacy budget) and in parentheses is the corresponding variance of the metrics.}
\label{npPrivconformalrf}
\end{table}

\subsection{Sample size effect with non-DP model}
\label{samplesizeNDP_RF}
\begin{table}[H]
\centering
\resizebox{\textwidth}{!}{%
\begin{tabular}{lcc|cc|cc}
\toprule
\textbf{Samples} 
& \multicolumn{2}{c|}{\textbf{Coverage}} 
& \multicolumn{2}{c|}{\textbf{Efficiency}} 
& \multicolumn{2}{c}{\textbf{Informativeness}} \\
\cmidrule(lr){2-3} \cmidrule(lr){4-5} \cmidrule(lr){6-7}
& \textbf{ExponQ} & \textbf{P-COQS} 
& \textbf{ExponQ} & \textbf{P-COQS} 
& \textbf{ExponQ} & \textbf{P-COQS} \\
\midrule
100  & 0.9892 (0.0326) & 0.9280 (0.0973) & 1.8658 (0.1716) & 1.5599 (0.3114) & 0.1342 (0.1716) & 0.4401 (0.3114) \\
200  & 0.9931 (0.0183) & 0.9298 (0.0678) & 1.8925 (0.1308) & 1.4749 (0.2301) & 0.1075 (0.1308) & 0.5251 (0.2301) \\
500  & 0.9973 (0.0076) & 0.9083 (0.0455) & 1.9274 (0.0872) & 1.3108 (0.1210) & 0.0726 (0.0872) & 0.6892 (0.1210) \\
1000 & 0.9987 (0.0036) & 0.9056 (0.0304) & 1.9520 (0.0577) & 1.2706 (0.0786) & 0.0480 (0.0577) & 0.7294 (0.0786) \\
2000 & 0.9865 (0.0126) & 0.9054 (0.0227) & 1.7402 (0.1492) & 1.2501 (0.0561) & 0.2598 (0.1492) & 0.7499 (0.0561) \\
6000 & 0.9393 (0.0181) & 0.9026 (0.0123) & 1.3745 (0.0986) & 1.2268 (0.0293) & 0.6255 (0.0986) & 0.7732 (0.0293) \\
10000 & 0.9229 (0.0120) & 0.9025 (0.0098) & 1.2939 (0.0406) & 1.2224 (0.0222) & 0.7061 (0.0406) & 0.7776 (0.0222) \\
\bottomrule
\end{tabular}%
}
\caption{Average effect of sample size on DP CP ( $\epsilon_{CP}=1$) with non-DP Random Forest model. The numbers are metric averages over $1000$ runs (per method and sample size) and in parentheses is the corresponding variance of the metrics. Average effect of sample size on private CP with Non-DP Random Forest model.}
\label{tab:random_forest_samples}
\end{table}

\subsection{Effect of $\alpha$ with non-DP model}\label{alpha_NP_RF}
\begin{table}[H]
\centering
\resizebox{\textwidth}{!}{%
\begin{tabular}{lcc|cc|cc}
\toprule
\textbf{$\alpha$} 
& \multicolumn{2}{c|}{\textbf{Coverage}} 
& \multicolumn{2}{c|}{\textbf{Efficiency}} 
& \multicolumn{2}{c}{\textbf{Informativeness}} \\
\cmidrule(lr){2-3} \cmidrule(lr){4-5} \cmidrule(lr){6-7}
& \textbf{ExponQ} & \textbf{P-COQS} 
& \textbf{ExponQ} & \textbf{P-COQS} 
& \textbf{ExponQ} & \textbf{P-COQS} \\
\midrule
0.01  & 0.9998 (0.0004) & 0.9909 (0.0033) & 1.9885 (0.0123) & 1.7554 (0.0370) & 0.0115 (0.0123) & 0.2446 (0.0370) \\
0.05  & 0.9800 (0.0120) & 0.9516 (0.0072) & 1.6512 (0.1426) & 1.4204 (0.0275) & 0.3488 (0.1426) & 0.5796 (0.0275) \\
0.10  & 0.9229 (0.0120) & 0.9025 (0.0098) & 1.2939 (0.0406) & 1.2224 (0.0222) & 0.7061 (0.0406) & 0.7776 (0.0222) \\
\bottomrule
\end{tabular}%
}
\caption{Average effect of $\alpha$ on DP CP ($\epsilon_{CP}=1$) with non-DP Random Forest model. The numbers are metric averages over $1000$ runs (per method and $\alpha$ value) and in parentheses is the corresponding variance of the metrics.}
\label{tab:random_forest}
\end{table}

\section{Simulation Results With DP Random Forest Model}
\label{app:DPrandom_forest}
The parameter settings are as in the DP NB model and the model's accuracy is $79\%$.
\subsection{Effect of $\epsilon_{CP}$ with DP model}
\label{prandomforest}

\begin{table}[H]
\centering
\resizebox{\textwidth}{!}{%
\begin{tabular}{lcc|cc|cc}
\toprule
\textbf{$\epsilon_{\text{conformal}}$} 
& \multicolumn{2}{c|}{\textbf{Coverage}} 
& \multicolumn{2}{c|}{\textbf{Efficiency}} 
& \multicolumn{2}{c}{\textbf{Informativeness}} \\
\cmidrule(lr){2-3} \cmidrule(lr){4-5} \cmidrule(lr){6-7}
& \textbf{ExponQ} & \textbf{P-COQS} 
& \textbf{ExponQ} & \textbf{P-COQS} 
& \textbf{ExponQ} & \textbf{P-COQS} \\
\midrule
0.1  & 1.0000 (0.0002) & 0.8982 (0.0112) & 1.9970 (0.0087) & 1.2552 (0.0438) & 0.0030 (0.0087) & 0.7448 (0.0438) \\
0.5  & 0.9575 (0.0272) & 0.8977 (0.0107) & 1.5621 (0.2378) & 1.2535 (0.0438) & 0.4379 (0.2378) & 0.7465 (0.0438) \\
1    & 0.9242 (0.0123) & 0.8977 (0.0106) & 1.3459 (0.0555) & 1.2533 (0.0438) & 0.6541 (0.0555) & 0.7467 (0.0438) \\
1.5  & 0.9171 (0.0105) & 0.8976 (0.0106) & 1.3183 (0.0465) & 1.2532 (0.0439) & 0.6817 (0.0465) & 0.7468 (0.0439) \\
3    & 0.9096 (0.0112) & 0.8977 (0.0106) & 1.2921 (0.0471) & 1.2534 (0.0442) & 0.7079 (0.0471) & 0.7466 (0.0442) \\
5    & 0.9058 (0.0121) & 0.8976 (0.0105) & 1.2801 (0.0505) & 1.2531 (0.0440) & 0.7199 (0.0505) & 0.7469 (0.0440) \\
10   & 0.9034 (0.0129) & 0.8975 (0.0104) & 1.2726 (0.0526) & 1.2530 (0.0438) & 0.7274 (0.0526) & 0.7470 (0.0438) \\
\bottomrule
\end{tabular}%
}
\caption{Average effect of CP privacy budget $\epsilon_{CP}$ on conformal prediction with DP Random Forest model and $\epsilon_{f}=2$. The numbers are metric averages over $1000$ runs (per method and privacy budget) and in parentheses is the corresponding variance of the metrics.}
\label{PrivconformalRF}
\end{table}

\subsection{Sample size effect with DP model}
\label{samplesizeRF}

\begin{table}[H]
\centering
\resizebox{\textwidth}{!}{%
\begin{tabular}{lcc|cc|cc|c}
\toprule
\textbf{Sample Size} 
& \multicolumn{2}{c|}{\textbf{Coverage}} 
& \multicolumn{2}{c|}{\textbf{Efficiency}} 
& \multicolumn{2}{c|}{\textbf{Informativeness}} 
& \textbf{Model Accuracy} \\
\cmidrule(lr){2-3} \cmidrule(lr){4-5} \cmidrule(lr){6-7}
& \textbf{ExponQ} & \textbf{P-COQS} 
& \textbf{ExponQ} & \textbf{P-COQS} 
& \textbf{ExponQ} & \textbf{P-COQS} 
& \\
\midrule
100    & 0.9979 (0.0142) & 0.9669 (0.0800) & 1.9977 (0.0146) & 1.9534 (0.1166) & 0.0022 (0.0146) & 0.0466 (0.1166) & 0.4201 (0.1141) \\
200    & 0.9993 (0.0060) & 0.9597 (0.0653) & 1.9990 (0.0078) & 1.9306 (0.1092) & 0.0010 (0.0078) & 0.0694 (0.1092) & 0.4503 (0.0975) \\
500    & 0.9998 (0.0019) & 0.9098 (0.0484) & 1.9983 (0.0094) & 1.6637 (0.1346) & 0.0017 (0.0094) & 0.3364 (0.1346) & 0.6183 (0.0776) \\
1000   & 0.9999 (0.0011) & 0.9060 (0.0321) & 1.9977 (0.0110) & 1.5422 (0.1067) & 0.0023 (0.0110) & 0.4578 (0.1067) & 0.6808 (0.0604) \\
2000   & 0.9978 (0.0065) & 0.9066 (0.0226) & 1.9678 (0.0756) & 1.4203 (0.0800) & 0.0322 (0.0756) & 0.5797 (0.0800) & 0.7371 (0.0455) \\
6000   & 0.9467 (0.0258) & 0.9009 (0.0133) & 1.5064 (0.1986) & 1.2913 (0.0519) & 0.4936 (0.1986) & 0.7087 (0.0519) & 0.7836 (0.0239) \\
10000  & 0.9242 (0.0123) & 0.8977 (0.0106) & 1.3459 (0.0555) & 1.2533 (0.0438) & 0.6541 (0.0555) & 0.7467 (0.0438) & 0.7941 (0.0173) \\
\bottomrule
\end{tabular}%
}
\caption{Average effect of sample size on DP CP ( $\epsilon_{CP}=1$) with DP Random Forest model ($\epsilon_{f}=2$). The numbers are metric averages over $1000$ runs (per method and sample size) and in parentheses is the corresponding variance of the metrics.}
\label{samplesizrf}
\end{table}

\subsection{Effect of $\alpha$ with DP model}
\label{alphaeffectRF}
\begin{table}[H]
\centering
\resizebox{\textwidth}{!}{%
\begin{tabular}{lcc|cc|cc}
\toprule
\textbf{$\alpha$} 
& \multicolumn{2}{c|}{\textbf{Coverage}} 
& \multicolumn{2}{c|}{\textbf{Efficiency}} 
& \multicolumn{2}{c}{\textbf{Informativeness}} \\
\cmidrule(lr){2-3} \cmidrule(lr){4-5} \cmidrule(lr){6-7}
& \textbf{ExponQ} & \textbf{P-COQS} 
& \textbf{ExponQ} & \textbf{P-COQS} 
& \textbf{ExponQ} & \textbf{P-COQS} \\
\midrule
0.01  & 1.0000 (0.0002) & 0.9904 (0.0036) & 1.9970 (0.0087) & 1.7604 (0.0422) & 0.0030 (0.0087) & 0.2396 (0.0422) \\
0.05  & 0.9914 (0.0121) & 0.9484 (0.0077) & 1.8486 (0.1751) & 1.4475 (0.0435) & 0.1514 (0.1751) & 0.5525 (0.0435) \\
0.10  & 0.9242 (0.0123) & 0.8977 (0.0106) & 1.3459 (0.0555) & 1.2533 (0.0438) & 0.6541 (0.0555) & 0.7467 (0.0438) \\
\bottomrule
\end{tabular}%
}
\caption{Average effect of $\alpha$ on DP CP ($\epsilon_{CP}=1$) with DP Random Forest model ($\epsilon_{f}=2$). The numbers are metric averages over $1000$ runs (per method and $\alpha$ value) and in parentheses is the corresponding variance of the metrics.}
\label{npPrivconformal}
\end{table}

\subsection{Average CP computational time with DP and non-DP models}\label{AveragetimeRF}

\begin{table}[H]
\centering
\resizebox{\textwidth}{!}{%
\begin{tabular}{lcccccc}
    \toprule
   Description & Empirical coverage & Efficiency & Informativeness & Model Accuracy & Ave. time (secs) \\
    \midrule
    %Npriv\_model and Npriv\_conform & 0.9004 (0.0098) & 1.1783 (0.0196) & 0.8217 (0.0196) & 0.8253 (0.0092) & - \\
    $\hat{f}$ and \textsc{ExponQ} & 0.9229 (0.0120) & 1.2939 (0.0406) & 0.7061 (0.0406) & 0.8125 (0.0093) & 0.7353 (0.0098)  \\
    $\hat{f}$ and P-COQS & 0.9025 (0.0098) & 1.2224 (0.0222) & 0.7776 (0.0222) & 0.8125 (0.0093) & 0.0074 (0.0002) \\
    %Priv\_model and Npriv\_conform & 0.8998 (0.0096) & 1.3981 (0.0318) & 0.6019 (0.0318) & 0.7479 (0.0125) & - \\
    $\hat{f}^{DP}$ and \textsc{ExponQ} & 0.9242 (0.0123) & 1.3459 (0.0555) & 0.6541 (0.0555) & 0.7941 (0.0173) & 0.7390 (0.0088) \\
    $\hat{f}^{DP}$ and P-COQS &0.8977 (0.0106) & 1.2533 (0.0443) & 0.7467 (0.0443) & 0.7941 (0.0173) & 0.0074 (0.0002) \\
    \bottomrule
\end{tabular}%
}
\caption{Average computational time of P-COQS and \textsc{ExponQ} combined with non-DP and DP versions of the Random Forest classifier. The numbers are metric averages over $1000$ runs (per setting) and in parentheses is the corresponding variance of the metrics.}
\label{averagetimeRF}
\end{table}

\subsection{Effect of Model Privacy Budget ($\epsilon_{f}$)}
\begin{table}[H]
\centering
\resizebox{\textwidth}{!}{%
\begin{tabular}{lcc|cc|cc|c}
\toprule
\textbf{$\epsilon_{f}$} 
& \multicolumn{2}{c|}{\textbf{Coverage}} 
& \multicolumn{2}{c|}{\textbf{Efficiency}} 
& \multicolumn{2}{c|}{\textbf{Informativeness}} 
& \textbf{Model Accuracy} \\
\cmidrule(lr){2-3} \cmidrule(lr){4-5} \cmidrule(lr){6-7}
& \textbf{ExponQ} & \textbf{P-COQS} 
& \textbf{ExponQ} & \textbf{P-COQS} 
& \textbf{ExponQ} & \textbf{P-COQS} 
& \\
\midrule
0.1  & 0.9256 (0.0142) & 0.9021 (0.0167) & 1.5647 (0.0841) & 1.4891 (0.0723) & 0.4353 (0.0841) & 0.5109 (0.0723) & 0.6870 (0.0591) \\
0.5  & 0.9248 (0.0126) & 0.8999 (0.0112) & 1.3839 (0.0674) & 1.2956 (0.0530) & 0.6161 (0.0674) & 0.7044 (0.0530) & 0.7803 (0.0242) \\
1    & 0.9241 (0.0127) & 0.8982 (0.0104) & 1.3575 (0.0611) & 1.2664 (0.0454) & 0.6425 (0.0611) & 0.7336 (0.0454) & 0.7898 (0.0187) \\
2    & 0.9242 (0.0123) & 0.8977 (0.0106) & 1.3459 (0.0555) & 1.2533 (0.0438) & 0.6541 (0.0555) & 0.7467 (0.0438) & 0.7941 (0.0173) \\
5    & 0.9241 (0.0123) & 0.8973 (0.0104) & 1.3405 (0.0546) & 1.2481 (0.0407) & 0.6595 (0.0546) & 0.7519 (0.0407) & 0.7957 (0.0167) \\
10   & 0.9241 (0.0122) & 0.8972 (0.0104) & 1.3405 (0.0542) & 1.2479 (0.0409) & 0.6595 (0.0542) & 0.7521 (0.0409) & 0.7958 (0.0168) \\
\bottomrule
\end{tabular}%
}
\caption{Average effect of model privacy budget $\epsilon_{f}$ on DP CP ($\epsilon_{CP}=1$) with DP Random Forest model. The numbers are metric averages over $1000$ runs (per method and privacy budget) and in parentheses is the corresponding variance of the metrics.}
\label{modelprivacyRF}
\end{table}

\section{Simulation Results With Non-DP Na\"ive Bayes Model}\label{app:Nonprivate_Naive_Bayes}
The parameter settings are as described in Section \ref{sec:simulations} and the accuracy of the model is $83\%$.
\subsection{Effect of Privacy Budget ($\epsilon_{CP}$)}\label{npnaivebayes}
\begin{table}[H]
\centering
\resizebox{\textwidth}{!}{%
\begin{tabular}{lcc|cc|cc}
\toprule
\textbf{$\epsilon_{CP}$} 
& \multicolumn{2}{c|}{\textbf{Coverage}} 
& \multicolumn{2}{c|}{\textbf{Efficiency}} 
& \multicolumn{2}{c}{\textbf{Informativeness}} \\
\cmidrule(lr){2-3} \cmidrule(lr){4-5} \cmidrule(lr){6-7}
& \textbf{ExponQ} & \textbf{P-COQS} 
& \textbf{ExponQ} & \textbf{P-COQS} 
& \textbf{ExponQ} & \textbf{P-COQS} \\
\midrule
0.1  & 0.9999 (0.0004) & 0.9005 (0.0104) & 1.9678 (0.0289) & 1.1787 (0.0223) & 0.0322 (0.0289) & 0.8213 (0.0223) \\
0.5  & 0.9444 (0.0174) & 0.9005 (0.0101) & 1.3465 (0.0972) & 1.1788 (0.0206) & 0.6535 (0.0972) & 0.8212 (0.0206) \\
1    & 0.9227 (0.0116) & 0.9006 (0.0099) & 1.2509 (0.0360) & 1.1788 (0.0201) & 0.7491 (0.0360) & 0.8212 (0.0201) \\
1.5  & 0.9154 (0.0103) & 0.9006 (0.0099) & 1.2254 (0.0267) & 1.1788 (0.0201) & 0.7746 (0.0267) & 0.8212 (0.0201) \\
3    & 0.9081 (0.0097) & 0.9006 (0.0099) & 1.2019 (0.0213) & 1.1788 (0.0200) & 0.7981 (0.0213) & 0.8212 (0.0200) \\
5    & 0.9051 (0.0097) & 0.9006 (0.0099) & 1.1925 (0.0203) & 1.1789 (0.0200) & 0.8075 (0.0203) & 0.8211 (0.0200) \\
10   & 0.9029 (0.0097) & 0.9006 (0.0099) & 1.1857 (0.0198) & 1.1789 (0.0201) & 0.8143 (0.0198) & 0.8211 (0.0201) \\
\bottomrule
\end{tabular}%
}
\caption{Average effect of CP privacy budget $\epsilon_{CP}$ on conformal prediction with non-DP Na\"ive Bayes model. The numbers are metric averages over $1000$ runs (per method and privacy budget) and in parentheses is the corresponding variance of the metrics.}
\label{npPrivconformalNB}
\end{table}

\subsection{Sample size effect}
\label{samplesizeNB}
\begin{table}[H]
\centering
\resizebox{\textwidth}{!}{%
\begin{tabular}{lcc|cc|cc}
\toprule
\textbf{Samples} 
& \multicolumn{2}{c|}{\textbf{Coverage}} 
& \multicolumn{2}{c|}{\textbf{Efficiency}} 
& \multicolumn{2}{c}{\textbf{Informativeness}} \\
\cmidrule(lr){2-3} \cmidrule(lr){4-5} \cmidrule(lr){6-7}
& \textbf{ExponQ} & \textbf{P-COQS} 
& \textbf{ExponQ} & \textbf{P-COQS} 
& \textbf{ExponQ} & \textbf{P-COQS} \\
\midrule
100  & 0.9794 (0.0434) & 0.9040 (0.1004) & 1.7300 (0.2179) & 1.3677 (0.2713) & 0.2700 (0.2179) & 0.6323 (0.2713) \\
200  & 0.9905 (0.0207) & 0.9206 (0.0695) & 1.7901 (0.1638) & 1.3394 (0.2004) & 0.2099 (0.1638) & 0.6606 (0.2004) \\
500  & 0.9967 (0.0078) & 0.9053 (0.0440) & 1.8566 (0.1158) & 1.2222 (0.1071) & 0.1434 (0.1158) & 0.7778 (0.1071) \\
1000 & 0.9984 (0.0042) & 0.9018 (0.0315) & 1.8951 (0.0876) & 1.1954 (0.0669) & 0.1049 (0.0876) & 0.8046 (0.0669) \\
2000 & 0.9836 (0.0129) & 0.9030 (0.0212) & 1.6272 (0.1483) & 1.1891 (0.0466) & 0.3728 (0.1483) & 0.8109 (0.0466) \\
6000 & 0.9382 (0.0169) & 0.9008 (0.0123) & 1.3149 (0.0795) & 1.1798 (0.0268) & 0.6851 (0.0795) & 0.8202 (0.0268) \\
10000 & 0.9227 (0.0116) & 0.9006 (0.0099) & 1.2509 (0.0360) & 1.1788 (0.0200) & 0.7491 (0.0360) & 0.8212 (0.0200) \\
\bottomrule
\end{tabular}%
}
\caption{Average effect of sample size on DP CP ($\epsilon_{CP}=1$) with non-DP Na\"ive Bayes model. The numbers are metric averages over $1000$ runs (per method and sample size) and in parentheses is the corresponding variance of the metrics. Average effect of sample size on private CP with Non-DP Random Forest model.}
\label{tab:naive_bayes_samples}
\end{table}

\subsection{Effect of $\alpha$}\label{alpha_NP_NB}
\begin{table}[H]
\centering
\resizebox{\textwidth}{!}{%
\begin{tabular}{lcc|cc|cc}
\toprule
\textbf{$\alpha$} 
& \multicolumn{2}{c|}{\textbf{Coverage}} 
& \multicolumn{2}{c|}{\textbf{Efficiency}} 
& \multicolumn{2}{c}{\textbf{Informativeness}} \\
\cmidrule(lr){2-3} \cmidrule(lr){4-5} \cmidrule(lr){6-7}
& \textbf{ExponQ} & \textbf{P-COQS} 
& \textbf{ExponQ} & \textbf{P-COQS} 
& \textbf{ExponQ} & \textbf{P-COQS} \\
\midrule
0.01  & 0.9999 (0.0004) & 0.9901 (0.0033) & 1.9677 (0.0289) & 1.6703 (0.0349) & 0.0323 (0.0289) & 0.3297 (0.0349) \\
0.05  & 0.9785 (0.0115) & 0.9500 (0.0073) & 1.5670 (0.1306) & 1.3610 (0.0244) & 0.4330 (0.1306) & 0.6390 (0.0244) \\
0.10  & 0.9227 (0.0116) & 0.9006 (0.0099) & 1.2509 (0.0360) & 1.1788 (0.0200) & 0.7491 (0.0360) & 0.8212 (0.0200) \\
\bottomrule
\end{tabular}%
}
\caption{Average effect of $\alpha$ on DP CP ($\epsilon_{CP}=1$) with non-DP Na\"ive Bayes model ($\epsilon_{f}=2$). The numbers are metric averages over $1000$ runs (per method and $\alpha$ value) and in parentheses is the corresponding variance of the metrics.}
\label{tab:naive_bayes}
\end{table}

\subsection{Effect of Model Privacy Budget ($\epsilon_{f}$)}

\begin{table}[H]
\centering
\resizebox{\textwidth}{!}{%
\begin{tabular}{lcc|cc|cc|c}
\toprule
\textbf{$\epsilon_{f}$} 
& \multicolumn{2}{c|}{\textbf{Coverage}} 
& \multicolumn{2}{c|}{\textbf{Efficiency}} 
& \multicolumn{2}{c|}{\textbf{Informativeness}} 
& \textbf{Model Accuracy} \\
\cmidrule(lr){2-3} \cmidrule(lr){4-5} \cmidrule(lr){6-7}
& \textbf{ExponQ} & \textbf{P-COQS} 
& \textbf{ExponQ} & \textbf{P-COQS} 
& \textbf{ExponQ} & \textbf{P-COQS} 
& \\
\midrule
0.1  & 0.9218 (0.0115) & 0.8994 (0.0101) & 1.6580 (0.0335) & 1.5905 (0.0233) & 0.3420 (0.0335) & 0.4095 (0.0233) & 0.6319 (0.0139) \\
0.5  & 0.9222 (0.0110) & 0.9000 (0.0094) & 1.5361 (0.0389) & 1.4532 (0.0263) & 0.4639 (0.0389) & 0.5468 (0.0263) & 0.7233 (0.0110) \\
1    & 0.9405 (0.0366) & 0.9812 (0.0409) & 1.7647 (0.1443) & 1.9207 (0.1718) & 0.2353 (0.1443) & 0.0793 (0.1718) & 0.6559 (0.0303) \\
2    & 0.9217 (0.0116) & 0.8999 (0.0098) & 1.4841 (0.0457) & 1.3985 (0.0322) & 0.5159 (0.0457) & 0.6015 (0.0322) & 0.7479 (0.0125) \\
5    & 0.9224 (0.0117) & 0.9000 (0.0098) & 1.2784 (0.0372) & 1.2036 (0.0216) & 0.7216 (0.0372) & 0.7964 (0.0216) & 0.8156 (0.0095) \\
10   & 0.9225 (0.0116) & 0.9002 (0.0098) & 1.2570 (0.0361) & 1.1840 (0.0204) & 0.7430 (0.0361) & 0.8160 (0.0204) & 0.8230 (0.0093) \\
\bottomrule
\end{tabular}%
}
\caption{Average effect of model privacy budget $\epsilon_{f}$ on conformal prediction with DP Na\"ive Bayes model. The numbers are metric averages over $1000$ runs (per method and privacy budget) and in parentheses is the corresponding variance of the metrics.}
\label{modelprivacyNB}
\end{table}

% \label{app:theorem}

% % Note: in this sample, the section number is hard-coded in. Following
% % proper LaTeX conventions, it should properly be coded as a reference:

% %In this appendix we prove the following theorem from
% %Section~\ref{sec:textree-generalization}:

% In this appendix we prove the following theorem from
% Section~6.2:

% \noindent
% {\bf Theorem} {\it Let $u,v,w$ be discrete variables such that $v, w$ do
% not co-occur with $u$ (i.e., $u\neq0\;\Rightarrow \;v=w=0$ in a given
% dataset $\dataset$). Let $N_{v0},N_{w0}$ be the number of data points for
% which $v=0, w=0$ respectively, and let $I_{uv},I_{uw}$ be the
% respective empirical mutual information values based on the sample
% $\dataset$. Then
% \[
% 	N_{v0} \;>\; N_{w0}\;\;\Rightarrow\;\;I_{uv} \;\leq\;I_{uw}
% \]
% with equality only if $u$ is identically 0.} \hfill\BlackBox

% \section{}

% \noindent
% {\bf Proof}. We use the notation:
% \[
% P_v(i) \;=\;\frac{N_v^i}{N},\;\;\;i \neq 0;\;\;\;
% P_{v0}\;\equiv\;P_v(0)\; = \;1 - \sum_{i\neq 0}P_v(i).
% \]
% These values represent the (empirical) probabilities of $v$
% taking value $i\neq 0$ and 0 respectively.  Entropies will be denoted
% by $H$. We aim to show that $\fracpartial{I_{uv}}{P_{v0}} < 0$....\\

% {\noindent \em Remainder omitted in this sample. See http://www.jmlr.org/papers/ for full paper.}

%\vskip 0.2in

\end{document}